\documentclass{amsart}
\usepackage{cite}
\usepackage{amsmath,amsfonts,amssymb,amsthm,amscd,mathrsfs,mathdots,bbm,color,bm,mathrsfs}
\numberwithin{equation}{section}
\usepackage{enumerate}
\usepackage{comment}
\usepackage{physics}
\usepackage{hyperref}
\usepackage[utf8]{inputenc}
\usepackage{graphicx}
\usepackage[linecolor=white,backgroundcolor=orange,bordercolor=red]{todonotes}
\usepackage{tikz}
\usetikzlibrary {positioning}
\usetikzlibrary{decorations.markings}
\usepackage{pgfplots}
\pgfplotsset{compat=1.14}
\definecolor {processblue}{cmyk}{0.96,0,0,0}

\tikzset{->-/.style={decoration={
  markings,
  mark=at position .5 with {\arrow{>}}},postaction={decorate}}}

\usepackage{newtxmath}
\usepackage{newtxtext}
\usepackage{subcaption}

\usepackage{ stmaryrd }

\DeclareFontFamily{U}{mathx}{}
\DeclareFontShape{U}{mathx}{m}{n}{<-> mathx10}{}
\DeclareSymbolFont{mathx}{U}{mathx}{m}{n}
\DeclareMathAccent{\widecheck}{0}{mathx}{"71}

\def\R{{\mathbb R}}
\def\C{{\mathbb C}}
\def\N{{\mathbb N}}
\def\Z{{\mathbb Z}}

\def\T{{\mathbb T}}

\def\rme{\mathrm{e}}

\def\ii{\mathrm{i}}

\def\sinc{\operatorname{sinc}}

\newtheorem{thm}{Theorem}

\newtheorem{prop}{Proposition}

\newtheorem{definition}{Definition}

\theoremstyle{definition}

\newtheorem{assumption}{Assumption}
\theoremstyle{remark}
\newtheorem{rem}{Remark}

\newcommand{\be}{\begin{equation}}
\newcommand{\ee}{\end{equation}}

\newcommand{\ben}{\begin{equation*}}
\newcommand{\een}{\end{equation*}}

\def\_#1{\def\next{#1}%
 \ifx\next\risingsign\expandafter\rising\els\rme^{\underline{#1}}\fi}
\def\risingsign{^}
\def\rising#1{^{\overline{#1}}}

\title{Quantum systems with jump-discontinous mass. I}

\author[F.~D.~Cunden]{Fabio Deelan Cunden}
\address{Dipartimento di Matematica, Univerisit\`a degli Studi di Bari, I-70125 Bari, Italy, and
INFN, Sezione di Bari, I-70126 Bari, Italy}
\email{fabio.cunden@uniba.it}
\author[G.~Gramegna]{Giovanni Gramegna}
\address{Dipartimento di Fisica, Univerisit\`a degli Studi di Bari, I-70126 Bari, Italy, and
INFN, Sezione di Bari, I-70126 Bari, Italy}
\email{giovanni.gramegna@uniba.it}
\author[M.~Ligab\`o]{Marilena Ligab\`o}
\address{Dipartimento di Matematica, Univerisit\`a degli Studi di Bari, I-70125 Bari, Italy}
\email{marilena.ligabo@uniba.it}

\begin{document}

\tikzset{midarrow/.style={decoration={markings, mark=at position 0.5 with {\arrow{stealth}}}, postaction={decorate}}}

	\maketitle
	
\begin{abstract}
We consider a free quantum particle in one dimension whose mass profile exhibits jump discontinuities. The corresponding Hamiltonian is a self-adjoint realisation of the kinetic-energy operator, with the specific realisation determined by the boundary conditions at the points of mass discontinuity. For a family of scale-free boundary conditions, we analyse the associated spectral problem. We find that the eigenfunctions exhibit a highly sensitive and erratic dependence on the energy.  Notably, the system supports infinitely many distinct semiclassical limits, each labeled by a point on a spectral curve embedded in the two-torus. These results demonstrate a rich interplay between discontinuous coefficients, boundary data, and spectral asymptotics.
\end{abstract}
\section{Introduction}

A free quantum particle of mass $m$, confined in an open bounded region $\Omega\subset\R^d$, $d \geq 1$, is formally described by the kinetic-energy operator,
\begin{equation}\label{eq:Kinetic}
H=-\frac{\hbar^2}{2m}\Delta,
\end{equation}
which acts on a proper subspace of the Hilbert space $L^2(\Omega)$. Equation (\ref{eq:Kinetic}) prescribes the action of $H$ only in the bulk of the system. The Hamiltonian $H$ should indeed be equipped with suitable boundary conditions (b.c.), specifying the behaviour of the particle at the boundary $\partial \Omega$, in order to generate a well-defined quantum dynamics. In quantum mechanics the 
 admissible b.c., encoded in the domain $D(H)$ of $H$, are constrained by the requirement that $H$ must be a self-adjoint operator, i.e. $D(H) = D(H^*)$ and $H = H^*$. Indeed, self-adjointness is a necessary and sufficient condition for a symmetric operator to have a purely real spectrum and to generate a unitary dynamics.
Different self-adjoint extensions correspond to different behaviours of the particle at the boundary,  they generate different dynamics and represent different physical situations.
As an example, we can consider the Dirichlet b.c.\ (i.e., we ask the wavefunctions to vanish on $\partial \Omega$). This is the prototypical setting of a \emph{quantum billiard}: a quantum counterpart of a classical particle of mass $m$ moving freely in $\Omega$ and bouncing off (specular reflection) at the boundary $\partial \Omega$.  It is known that there is
a non-decreasing sequence $0<E_1\leq E_2\leq\cdots$ of eigenvalues accumulating at $+\infty$, and an orthonormal basis $\{\psi_{E_{\kappa}}\}_{\kappa\geq1}$ in $H^2(\Omega)\cap H^1_0(\Omega)$ such that 
\be
H\psi_{E_{\kappa}}=E_{\kappa}\psi_{E_{\kappa}}, \quad \textrm{for all $\kappa \geq 1$.}
\ee
One of the most explored questions on quantum billiards regards the eigenfunctions delocalisations in the high-frequency limit $\kappa\to+\infty$ (here equivalent to the semiclassical limit $\hbar\to0$ if we impose that the energy levels converge to a classical limit $E_{\kappa}\to E$, as $ \kappa \to + \infty$). In this regime, it is expected that
quantum mechanics should converge to classical mechanics in a certain sense. 
The quantum mechanical interpretation of the measure $|\psi_{E_{\kappa}}|^2dx$ is  the probability law of the position of a particle with defined energy $E_{\kappa}$.  A major theme in quantum chaos is understanding which measures $\mu$ can arise as
weak-$^*$ limits of $|\psi_{E_{\kappa}}|^2dx$ in the high-frequency limit; this includes the Quantum Ergodicity theorem
and the Quantum Unique Ergodicity conjecture. See~\cite{Anantharaman22} for an updated survey of results and methods.

There is a vast literature on the subject  in dimension $d\geq2$. The one-dimensional case $d=1$ is hardly discussed in the literature and in textbooks (even the most pedagogical). To see why, consider a quantum particle of mass $m$ in a one-dimensional box $\Omega=(-\ell/2,\ell/2)$. We may say that the Hamiltonian of the system is 
$$
H=-\frac{\hbar^2}{2m}\frac{d^2}{d x^2},
$$
with domain $H^2_0(\Omega)$ (Dirichlet b.c.).  Its eigenvalues and eigenfunctions are textbook:
\begin{equation}
E_{\kappa}=\frac{\hbar^2 \pi^2 \kappa^2}{2m\ell^2},\quad \psi_{E_{\kappa}}(x)=\sqrt{\frac{2}{\ell}}\sin\left(\dfrac{\kappa\pi}{\ell}\left(x+\frac{\ell}{2}\right)\right)\chi_{(-\ell/2,\ell/2)}(x), \quad \kappa \geq 1.
\end{equation}
In the high-frequency limit $\kappa \to +\infty$, the probability densities $|\psi_{E_{\kappa}}|^2 dx$ flatten out, converging to the uniform measure:
\begin{equation}
d\mu= \frac{1}{|\Omega|}\chi_{\Omega}(x)dx.
\end{equation}
All self-adjoint realisations of the Laplacian in dimension $d=1$ yield this same macroscopic density. While the full phase-space distributions (e.g., Wigner functions) may oscillate wildly and resist convergence, their spatial projections always settle to the uniform measure. For a detailed discussion, see~\cite{Angelone2024}. The universality of the microscopic density is discussed in~\cite{Cunden18}.

Because of this apparent regularity, one-dimensional quantum billiards are often dismissed as uninteresting  (as one-dimensional autonomous systems are completely integrable). After all, playing billiard on a line requires no skills at all! 
Here we challenge this perspective by showing that certain one-dimensional quantum billiards exhibit a richness of structure and behaviour worthy of close attention. 
Their spectral properties and eigenfunction statistics can be as instructive as their higher-dimensional cousins if not more so, due to their tractability. 

\subsection{A glimpse of the setting and results}
\label{sec:glimpse}
The system considered in this paper is a quantum particle confined not to a single interval, but to the union of two, forming a piecewise domain
\begin{equation}
\Omega=I_1\cup I_2 , \quad I_1=(-\ell_1,0),\quad
I_2=(0,\ell_2),
\end{equation}
and the corresponding Hilbert space is $L^2(\Omega)$. Different ways of `gluing' the billiards correspond to different b.c.\ on the kinetic energy operator.  
The particle evolves in this pair of joined billiards under the sole influence of its kinetic energy, yet with a crucial twist: it experiences a discontinuous change in mass. Specifically, the mass is constant and equal to $m_1$ on the left interval $I_1$, and $m_2$ on the right $I_2$, introducing  a discontinuity analogous to having  
different \emph{refractive indices} into an otherwise free system. 
We will be interested in the case  $m_1\neq m_2$, so that formally the Hamiltonian is the kinetic energy operator 
\begin{equation}
\label{eq:formalH}
H=-\frac{\hbar^2}{2m}\frac{d^2}{d x^2}, \quad \text{with}
    \quad m=
    \begin{cases}
m_1 &\text{for $x\in I_1$}\\
m_2 &\text{for $x\in I_2$}
    \end{cases}.
\end{equation}
The Hamiltonian 
 has a tauntingly simple form, but exhibits rich and complex behaviour.

 A possible measure of the delocalisation of  the wavefunctions, is the \emph{leaning}~\cite{Ares2020} defined, for any nonzero $\psi \in L^2(\Omega)$, as
    the ratio
   \begin{equation}
      \mathcal{L}(\psi)= \frac{\int_{I_2}|\psi(x)|^2\, dx -\int_{I_1}|\psi(x)|^2 \, dx}{\int_{I_2}|\psi(x)|^2 \, dx+\int_{I_1}|\psi(x)|^2\, dx}\in[-1,1],
   \end{equation}
which records whether the quantum probability is skewed towards the left or right subintervals. Negative (resp.\ positive) values of $\mathcal{L}(\psi)$  correspond to densities that `lean' on the left (resp.\ right).
   The extreme values $-1$ and $1$ correspond to wavefunctions completely localised in $I_1$ or $I_2$, respectively.
   
Consider for instance the eigenfunctions of~\eqref{eq:formalH} that  vanish
(Dirichlet b.c.) at the edges $\psi(-\ell_1)=\psi(\ell_2)=0$ and that satisfy the zero current and continuity conditions (Kirchhoff b.c.) at the junction $\frac{1}{m_1}\psi'(0^-)=\frac{1}{m_2}\psi'(0^+)$, $\psi(0^{-})=\psi(0^{+})$.
 When the masses are equal ($m_1=m_2$), the problem  reduces to the textbook scenario of a single interval with Dirichlet boundary conditions. In this case, the eigenfunctions $\psi_{E_{\kappa}}$ become asymptotically equidistributed, and one recovers the classical ergodic time average: 
\begin{equation}
  \lim_{\kappa\rightarrow +\infty}\mathcal{L}(\psi_{E_{\kappa}})=\frac{\ell_2-\ell_1}{\ell_2+\ell_1}.
\end{equation}
Classically, this matches the difference between the fraction of time spent in the two intervals by a particle bounding back and forth between $-\ell_1$ and $\ell_2$ with constant speed. 
\begin{figure}
    \centering
    \includegraphics[width=1\linewidth]{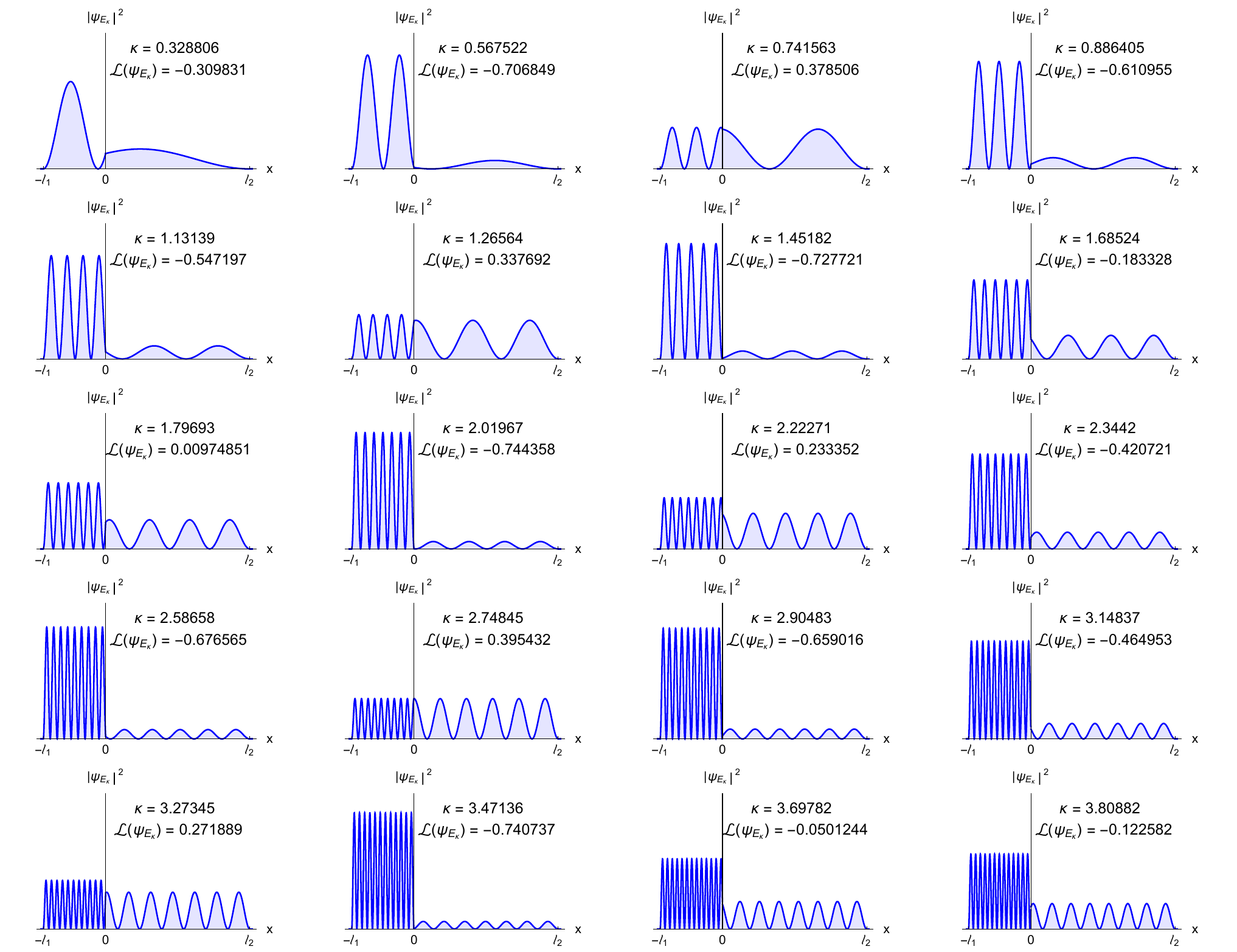}
    \caption{Particle density $|\psi_{E_{\kappa}}|^2$ for the first $20$ values of $E_{\kappa}=\hbar^2\kappa^2/2$.  Here $m_1=16$, $m_2=1$, $\ell_1=\rme$, and $\ell_2=\pi$. The corresponding frequencies $\omega_1=4\rme$ and $\omega_2=\pi$ are nonresonant. In the plots we also indicated the corresponding values of $\kappa$ and $\mathcal{L}(\psi_{E_{\kappa}})$.}
    \label{fig:eigenfunctions_leaning}
\end{figure}
The picture transforms dramatically when $m_1\neq m_2$. The particle density no longer exhibits a regular limiting behaviour, see 
 Fig.~\ref{fig:eigenfunctions_leaning}; the leaning becomes erratic, dependent on the arithmetic of the \emph{frequencies} of the system
\begin{equation}
    \omega_1=\sqrt{m_1}\ell_1,  \quad \omega_2=\sqrt{m_2}\ell_2.
\end{equation}
In the generic (nonresonant) case, the set of values $\{\mathcal{L}(\psi_{E_{\kappa}})\}_{{\kappa}}$ `fill' a proper interval, see  Fig.~\ref{fig:leaning_Kirk}.
Yet, amid this irregularity, order emerges in the form of upper limit, Ces\`aro limit, and lower limit:
\begin{align}
\label{eq:upper_limit}
\limsup_{\kappa \to + \infty} \mathcal{L}(\psi_{E_{\kappa}})&=\frac{\ell_2-\ell_1}{\ell_2+\ell_1},\\
\label{eq:Cesaro_limit}
    \lim_{E\to+\infty}\frac{1}{\#\{E_{\kappa}\leq E\}}\sum_{E_{\kappa}\leq E}\mathcal{L}(\psi_{E_{\kappa}})&=\frac{\sqrt{m_2}\ell_2-\sqrt{m_1}\ell_1}{\sqrt{m_2}\ell_2+\sqrt{m_1}\ell_1},\\
    \label{eq:lower_limit}
\liminf_{\kappa \to + \infty} \mathcal{L}(\psi_{E_{\kappa}})&=\frac{m_2\ell_2-m_1\ell_1}{m_2\ell_2+m_1\ell_1},
\end{align}
if $m_2\leq m_1$ (otherwise the values of the upper and lower limits are swapped), see Sec.~\ref{sec:scalefree} for a detailed description of this phenomenon. Note that when $m_1=m_2$ the limits~\eqref{eq:upper_limit}-\eqref{eq:Cesaro_limit}-\eqref{eq:lower_limit} all coincide. 

\begin{figure}
    \centering
    \includegraphics[width=.75\linewidth]{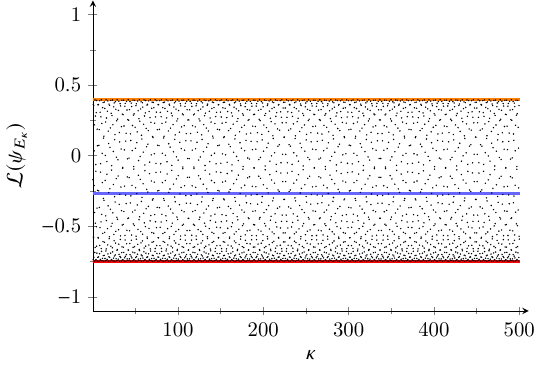}
    \caption{Leaning $\mathcal{L}(\psi_{E_{\kappa}})$ versus $\kappa$. We also display the upper limit~\eqref{eq:upper_limit} (orange line), the Ces\`aro  limit~\eqref{eq:Cesaro_limit} (blue  line), and the lower limit~\eqref{eq:lower_limit} (red line). Here the frequencies are nonresonant (same mass and length values as in Fig.~\ref{fig:eigenfunctions_leaning}).}
    \label{fig:leaning_Kirk}
\end{figure}

\subsection{Motivation and some related literature}

Around 2020, there were reports by Ares, Esteve,  Falceto, and Us\'on~\cite{Ares2020} of numerical diagonalisation of two coupled free fermionic chains (hence each integrable on its own) with different hopping rates whose one-particle eigenstates have spatial distribution that depends on the energy in a rather erratic way. In the limit of long chains, the particle is confined in one of the two parts or, delocalised along the whole chain. There is no sharp localisation-delocalisation transition: Instead, there emerged a peculiar sensitivity of spatial distribution to the energy, with irregular, almost chaotic features. 

The present work originates as an attempt to detect whether such phenomena persist in continuous, rather than lattice, systems. The idea was that in such a infinite-dimensional model the problem of taking the limit of `infinitely long chains' would have been completely bypassed. We thus replace free fermionic chains by free particles (kinetic-operator) on segments of the real line. Gluing of two discrete chains with different hopping rates, corresponds to joining two intervals in which the particle would have two different masses. This translation of the lattice model into a continuous one came with several difficulties: 
\begin{enumerate}
	\item We faced the problem of defining a symmetric kinetic energy operator with position-dependent mass $m(x)$ that is already a non-obvious task (see Sec.~\ref{sec:PDM} below). In the present setting the mass has jump-discontinuities, and we overcame this obstruction by defining the model on the direct sum of two $L^2$ spaces.
	\item One can observe that the model of two joined one-dimensional billiards is a two-edges quantum graph. One can therefore use the classification of self-adjoint extensions of kinetic-operators on quantum graphs in terms of unitary operators~\cite{Asorey05,Bruening08,Asorey15} as well as in terms of scattering operators~\cite{Kostrykin99,Kostrykin00,Bolte09}. 
    However, the standard treatments assume equal mass on each edge. Incorporating different masses on distinct edges requires adapting the existing theory. 
	\item  There is substantial work on the spectral and
scattering properties of quantum graphs; see, for instance, the monograph~\cite{Berkolaiko2013}.  Notably, Kaplan~\cite{Kaplan01} analysed the detailed wave function structure on quantum ring graphs with \emph{random} edge lengths and reported deviations from the uniform (ergodic) distributions. Shortly after, Berkolaiko, Keating and Winn~\cite{Berkolaiko2003,Berkolaiko2004}  proved non-quantum-ergodicity in star graphs with a \emph{large number of edges}.  Our setting departs from both of these: the number of edges is fixed (only two) and there is no disorder (the lengths of the edges are not random). Here it is the jump discontinuity of the masses that breaks ergodicity. 
 \item In~\cite{Jakobson2015}  Jakobson, Safarov, Strohmaier and Colin de Verdi\`ere  elucidated the semiclassical limit of spectral theory on manifolds whose metrics have jump-like discontinuities along hypersurfaces, providing examples in dimension $d=2$ with zero Kirchhoff b.c.\ at the points of discontinuity.  Specifically, they related the semiclassical theory of discontinuous quantum systems to \emph{ray-splitting billiards} -- classical systems in which a single ray may bifurcate upon encountering a discontinuity.
One-dimensional models of quantum particle in a piecewise constant potential were explored by Bl\"umel et al.~\cite{Blumel2001,Blumel2002}  as a toy model to argue that ray-splitting billiards are equivalent to systems where \emph{newtonian} and \emph{non-newtonian trajectories} coexist. Something similar seems to happen in our setting of discontinuous mass (see below). 
We stress however that the sequence of particle densities $|\psi_{E_{\kappa}}|^2dx$ for a quantum particle in an interval with jump-discontinuos potential (and constant mass) has a semiclassical limit (no longer uniform in the interval $\Omega$), and it is therefore different from the quantum particle with zero (constant) potential and jump-discontinuous mass considered here which has not a semiclassical limit. It is our hope that our results 
will serve as a motivation to further investigation on quantum discontinuous systems and their semiclassical limits.
\end{enumerate}

\subsection{Outline of the paper}  In Sec.~\ref{sec:PDM} we define the model of a quantum particle with jump-discontinuous mass. We first consider the general problem of defining a one-dimensional Schr\"odinger operator with position-dependent mass $m(x)$. The nature of the problem forces us to consider separately the cases when i) the mass profile $m(x)$ is a regular function of the position  (Sec.~\ref{sub:regular}), and  ii)   $m(x)$ is piecewise regular (Sec.~\ref{sub:piecewise}). The latter case indicates a clear path to define our model of a particle with piecewise constant mass using the standard theory of self-adjoint extensions in terms of boundary conditions  (Proposition~\ref{prop:sa}); we introduce and characterise the class of scale-free (aka dilation-invariant) self-adjoint extensions of the kinetic energy operator. The connection to some physics literature on quantum systems with position-dependent mass is briefly presented in Sec.~\ref{sub:physics}.  In Sec.~\ref{sec:spectral_problem} we consider the spectral properties of the model. We first setup the general scheme in terms of spectral matrices and spectral functions in Sec.~\ref{sec:spectral_gen}. Then, in Sec.~\ref{sec:tori} we consider in detail the case of scale-free boundary conditions for which the spectral problem can be cast in the language of linear flows on the torus (Theorem~\ref{thm:tori}). As a warm-up example, in Sec~\ref{sec:uncoupled} we consider the (easy but instructive) spectral problem for two uncoupled intervals using the language of spectral matrices and spectral functions, and show that under a nonresonant condition, the eigenstates are completely localised.
Sec.~\ref{sec:scalefree} contains new explicit solutions to the spectral problem for two-edges quantum graphs (with distinct masses) with zero Kirchhoff b.c.\ at the connected vertices. In Sec.~\ref{sec:semiclassical} we consider the Wigner function in phase space of the eigenfunctions (Proposition~\ref{prop:Wigner}), and we show that for scale-free boundary conditions, there are infinitely many semiclassical limit of the Wigner functions, each of them parametrised by a point on the 2-torus. In the conclusions, Sec.~\ref{sec:conclusions}, we outline some future directions.

\section{One-dimensional Schr\"odinger operators with position dependent mass}
\label{sec:PDM}
To define our quantum model with discontinuous mass, we first revisit the notion of Schr\"odinger operators with position-dependent mass.
Consider the Hamiltonian of a classical system with position-dependent mass (PDM),
\begin{equation}
  \label{eq:classical_hamiltonian}
\mathcal{H}(x,p)=\frac{p^2}{2m(x)}+V(x),\quad (x,p)\in\R^2,
\end{equation}
with $m(x) >0$ and $V(x) \in \R$.
The corresponding quantization is given by an operator 
 on $L^2(\R)$  that is formally 
\begin{equation}
\frac{\hat{p}^2}{2m(\hat{x})}+V(\hat{x}),
\end{equation}
where $\hat{x}\psi(x)=x \psi(x)$ is the position operator and $\hat{p}\psi(x)=-\ii\hbar \psi'(x)$ is the momentum operator.
Notice that, since $\hat{p}$ and $\hat{x}$ do not commute, the  expression $\frac{\hat{p}^2}{2m(\hat{x})}$ is not well defined! A more suitable quantum version of~\eqref{eq:classical_hamiltonian} is obtained by writing the kinetic term in \emph{divergence form},
\begin{equation}
\label{eq:divergence}
H=\hat{p}\frac{1}{2m(\hat{x})}\hat{p}+V(\hat{x}),
\end{equation}
acting as
\begin{equation}
H\psi(x)=\left[\frac{\hbar}{\ii}\frac{d}{d x}\frac{1}{2m(x)}\frac{\hbar}{\ii}\frac{d}{d x}+V(x)\right]\psi(x)
\end{equation}
on normalised wavefunctions $\psi\in C^{\infty}_c(\R)$. The operator $H$ in~\eqref{eq:divergence} is a Sturm-Liouville operator and is evidently symmetric. 

\subsection{Regular mass profile}
\label{sub:regular}
If $m(x)$ is regular, then $H$ in~\eqref{eq:divergence} is unitarily equivalent to a Schr\"odinger operator with constant (unit) mass in a modified potential.

Consider the unitary transformation $W\colon L^2(\R)\to L^2(\R)$ given by
\begin{equation}
(W\psi)(x)=\psi(g(x))\sqrt{g'(x)}, 
\end{equation}
where $g(x)$ is a primitive of $\sqrt{m(x)}$.
Then, 
\begin{equation}
(W^*HW\varphi)(y)=\left[-\frac{\hbar^2}{2}\frac{\partial^2}{\partial y^2}+V_m(y)\right]\varphi(y),
\end{equation}
where the modified potential $V_m(x)$ is related to $V(x)$ and $m(x)$ by 
\begin{equation}
\label{eq:mod_pot}
V_m(g(x))=V(g(x))-\frac{\hbar^2}{8}\frac{m''(x)}{m^2(x)}+\frac{7\hbar^2}{32}\frac{(m'(x))^2}{m^3(x)}.
\end{equation}
The reader is referred to~\cite{Bagchi04} for a unified approach to many exact results known in the literature of PDM systems.

\subsection{Piecewise-regular mass} 
\label{sub:piecewise}
The above construction fails when $m(x)$ is not smooth everywhere.  Indeed, if $m(x)$ is not $C^2$, the modified potential~\eqref{eq:mod_pot} is singular. 

This obstruction appears already at the classical level. Hamilton's equations for the PDM classical system~\eqref{eq:classical_hamiltonian} are
\begin{equation}
    \begin{cases}
        \dot{p}=-\dfrac{\partial \mathcal{H}}{\partial x}=\dfrac{p^2}{2m^2(x)}m'(x)-V'(x)\\\\
        \dot{x}=\dfrac{\partial  \mathcal{H}}{\partial p}=\dfrac{p}{m(x)}.
    \end{cases}
\end{equation}
We see that if $m(x)$ in~\eqref{eq:classical_hamiltonian} is not  differentiable, then Hamilton's equations are not defined (the Hamiltonian vector field associated to $\mathcal{H}$ is not smooth).

This raises the question of how to define a quantum system with a piecewise regular mass profile $m(x)$. A rather natural way to proceed is as follows. 
\par

Let $m\colon (a,b)\to \mathbb{R}$ be a mass profile defined on the (not necessarily bounded) open interval $(a,b)\subset\mathbb{R}$. Suppose that $m(x)$ is a piecewise regular mass profile, i.e., there exists a finite partition  $a=\alpha_0< \alpha_1<\cdots<\alpha_n=b$ of  $(a,b)$ such that $m$ and its derivative $m'$ are defined and continuous in $(\alpha_0,\alpha_1)\cup(\alpha_1,\alpha_2)\cup\cdots\cup (\alpha_{n-1},\alpha_n)$. In such a case, we can first define the symmetric PDM operator $H=\hat{p}\frac{1}{2m(\hat{x})}\hat{p}+V(\hat{x})$ acting on functions $\psi\in C^{\infty}_c\left((\alpha_0,\alpha_1)\cup(\alpha_1,\alpha_2)\cup\cdots\cup (\alpha_{n-1},\alpha_n)\right)$ vanishing at the discontinuity points $\alpha_0,\alpha_1,\ldots,\alpha_n$, and then consider the possible self-adjoint extensions of $H$ in $L^2(a,b)$.
\par
\subsection{A quantum particle in a box with jump-discontinuous mass}
\label{sub:model}
We now implement the above strategy in the simplest case: a free quantum particle (zero potential) in an interval whose \emph{mass profile $m(x)$ is piecewise-constant with a single jump-discontinuity}.

Consider a quantum particle in $\Omega=(-\ell_1,0)\cup (0,\ell_2)\equiv I_1\cup I_2$ with mass $m_1>0$ in the segment $I_1$, and mass $m_2>0$ in the segment $I_2$.

Consider on the space $L^2(\Omega)\simeq L^2(I_1)\oplus L^2(I_1)$ the linear operator $H_{\operatorname{min}}$, 
acting as
\begin{equation}
\label{eq:Hmin}
H_{\operatorname{min}}\psi_1\oplus\psi_2=\left(-\frac{\hbar^2}{2m_1}\psi_1''\right)\oplus \left(-\frac{\hbar^2}{2m_2}\psi_2''\right).
\end{equation}
with domain
\begin{equation}
D(H_{\operatorname{min}})=H_0^2(I_1)\oplus H_0^2(I_2),
\end{equation}
where $H^2_0(\Omega)$ denotes the closure in the Sobolev space $H^2(\Omega)$ of infinitely differentiable compactly supported functions $C^\infty_c(\Omega)$.
The adjoint operator $H^*_{\operatorname{min}}$ has the same functional form of~\eqref{eq:Hmin} but is defined on the larger domain $D(H^*_{\operatorname{min}})=H^2(I_1)\oplus H^2(I_2)$. Therefore, $H_{\operatorname{min}}$ is symmetric but not self-adjoint. A direct calculation shows that the deficiency indices $n_{\pm}=\operatorname{dim}\ker \left(H^*_{\operatorname{min}}\mp \ii I\right)$ are $n_+=n_-=4$; therefore, by von Neumann's theory (see~\cite[Ch. X]{ReedSimon}) the operator $H_{\operatorname{min}}$ admits infinitely many self-adjoint extensions in one-to-one correspondence with unitary matrices $U\in\mathcal{U}(4)$. Each unitary matrix encodes the so-called \emph{connection rules} for the boundary data, that is, the values of $\psi$ and  $(1/m)\psi'$ at the boundary points of $\Omega$. Indeed, we recall~\cite{Asorey05,Asorey15} that  
the `lack of self-adjointess' of $H^*_{\operatorname{min}}$ is measured by noting that for all $\varphi,\psi\in H^2(I_1)\oplus H^2(I_2)$,
\begin{equation}
\label{eq:boundary_term}
    \langle \varphi, H^*_{\operatorname{min}}\psi\rangle- \langle H^*_{\operatorname{min}} \varphi, \psi\rangle=
    \frac{\hbar^2}{2}\left[\overline{\varphi_1}\frac{1}{m_1}\psi_1'-\overline{\varphi_1'}\frac{1}{m_1}\psi_1\right]^{x=0^-}_{x=-\ell_1}+\frac{\hbar^2}{2}\left[\overline{\varphi_2}\frac{1}{m_2}\psi_2'-\overline{\varphi_2'}\frac{1}{m_2}\psi_2\right]^{x=\ell_2}_{x=0^+}.
\end{equation}
(The boundary terms are well defined by the Sobolev embedding $H^2(I_1)\oplus H^2(I_2)\subset C^{1}(I_1)\oplus C^1(I_2)$.) Hence, $H^*_{\operatorname{min}}$  is self-adjoint if restricted 
 (equivalently,  $H_{\operatorname{min}}$  is self-adjoint if extended) to the subspace of  wavefunctions for which the boundary term in~\eqref{eq:boundary_term} vanishes.

Denote the boundary data of a wavefunction $\psi\in H^2(I_1)\oplus H^2(I_2)$ by
\begin{align}
\gamma{\psi}&=\left(\psi_1(-\ell_1),\psi_1(0^-),\psi_2(0^+),\psi_2(\ell_2)\right)^T,\\
\nu{\psi}&=
\left(
-\frac{1}{m_1}\psi_1'(-\ell_1),
\frac{1}{m_1}\psi_1'(0^-),
-\frac{1}{m_2}\psi_2'(0^+),
\frac{1}{m_2}\psi_2'(\ell_2)
\right)^T.
\end{align}
The following result is adapted from the classification of self-adjoint extensions of the Laplacian operator on general bounded domain in~\cite{Facchi18}.
 \begin{prop}
 \label{prop:sa} 
 The set of all self-adjoint extensions of $H_{\operatorname{min}}$ is
 \begin{equation}
     \left\{H_U : D(H_U) \to L^2(I_1) \oplus L^2(I_2)   \colon U\in \mathcal{U}(4)\right\},
 \end{equation}
  where for all $U\in \mathcal{U}(4)$, 
\begin{equation}
\label{eq:Dom_H_U}
D(H_U)=\left\{\psi\in H^2(I_1)\oplus H^2(I_2)\colon \ii(\mathbb{I}+U)\gamma{\psi} =(\mathbb{I}-U)\nu{\psi}\right\},
\end{equation}   
and $\mathbb{I}$ is the identity matrix.
 \end{prop}
For example,  $U=\mathbb{I}$ and $U=-\mathbb{I}$ correspond to Dirichlet b.c.\ and Neumann b.c., respectively ($\psi(x)=0$ and $\psi'(x)=0$ at the boundary points, respectively).
 \par

A special class of b.c.\ are those that are invariant under dilations, or \emph{scale-free} b.c.. Those are exactly the b.c.\ that do not `mix' the values of $\psi$ and $\psi'$ at the boundary. The next definition makes this condition precise.
\begin{definition}[Scale-free extensions and boundary conditions]\label{def:ScaleFree}
    We say that $H_U$ is a scale-free extension of $H_{\operatorname{min}}$, or that the corresponding boundary conditions are scale-free, if $U\in\mathcal{U}(4)$ satisfies $\operatorname{Ran}(\mathbb{I}+U)\cap\operatorname{Ran}(\mathbb{I}-U)=\{0\}$. In this case, the domain of $H_U$ defined in~\eqref{eq:Dom_H_U}, is given by the conditions \begin{equation}
        \gamma\psi\in\ker(\mathbb{I}+U),\quad \nu\psi\in\ker(\mathbb{I}-U),
    \end{equation} 
    or, equivalently, 
    \begin{equation}
    \gamma\psi\in\operatorname{Ran}(\mathbb{I}-U),\quad \nu\psi\in\operatorname{Ran}(\mathbb{I}+U).
        \end{equation} 
\end{definition}
The following proposition characterises the scale-free self-adjoint extensions $H_U$.
\begin{prop}\label{prop:ScaleFree} Let $U$ be a unitary matrix. The following are equivalent:
    \begin{enumerate}
        \item $\operatorname{Ran}(\mathbb{I}+U)\cap\operatorname{Ran}(\mathbb{I}-U)=\{0\}$;
        \item $\sigma(U)\subset\{-1,1\}$;
        \item $U=\mathbb{I}-2P$, with $P=P^2=P^*$ orthogonal projection.
    \end{enumerate}
\end{prop}
\begin{proof}
    We first prove $(1)\Rightarrow (2)$. Suppose that $v\neq 0$ is an eigenvector of $U$ with eigenvalue $\lambda$. From $Uv=\lambda v$, we have
    \begin{align*}
        (\mathbb{I}+U)(\mathbb{I}-U)v&=(\mathbb{I}+U)(1-\lambda)v=(1-\lambda^2)v\\
        (\mathbb{I}-U)(\mathbb{I}+U)v&=(\mathbb{I}-U)(1+\lambda)v=(1-\lambda^2)v.
    \end{align*}
    We conclude that $(1-\lambda^2)v\in \operatorname{Ran}(\mathbb{I}+U)\cap\operatorname{Ran}(\mathbb{I}-U)$. This implies $(1-\lambda^2)v=0$, namely $\lambda^2=1$.  
    \par
  We now show $(2)\Rightarrow (3)$. By the spectral theorem for normal matrices, $U=Q-P$, where $Q$ (respectively $P$) is the orthogonal projection onto the eigenspace relative to the eigenvalue $-1$ (respectively $+1$). Use $Q+P=I$ to conclude.
   \par
It remains to show $(3)\Rightarrow (1)$. Suppose $U=\mathbb{I}-2P$, and denote $Q=\mathbb{I}-P$. We have $\mathbb{I}+U=2(\mathbb{I}-P)=2Q$ and $\mathbb{I}-U=2P$. Since $PQ=QP=0$, we have that the $\operatorname{Ran}(\mathbb{I}+U)$ and $\operatorname{Ran}(\mathbb{I}-U)$ are orthogonal. 
 \end{proof}

Dirichlet ($U=\mathbb{I}$) and Neumann ($U=-\mathbb{I}$) are special cases of scale-free b.c.. General Robin b.c.\ are not scale-free.
\par
The scale-free self-adjoint extensions are nonnegative.
\begin{prop}
\label{prop:nonnegative}
    Let $H_U$, with $U=\mathbb{I}-2P$, $P$ orthogonal projection. Then, $H_U\geq0$.
\end{prop}
\begin{proof} 
Let $\psi\in D(H_U)$. Integrating by parts, 
\begin{equation}
    \langle \psi,H_U\psi\rangle_{L^2(\Omega)}=\int_{I_1}\frac{\hbar^2}{2m_1}|\psi'(x)|^2\, dx+\int_{I_2}\frac{\hbar^2}{2m_2}|\psi'(x)|^2\, dx-\frac{\hbar^2}{2}\langle\gamma\psi,\nu\psi\rangle_{\mathbb{C}^4}.
\end{equation}
When $U=\mathbb{I}-2P$ the condition defining $D(H_U)$ in~\eqref{eq:Dom_H_U} reads $\ii(\mathbb{I}-P)\gamma\psi=P\nu\psi$; since $P$ and $\mathbb{I}-P$ have orthogonal ranges, we have 
\begin{equation}
\label{eq:orthogonality_scalefree}
\langle\gamma\psi,\nu\psi\rangle_{\mathbb{C}^4}=0.
\end{equation}
\end{proof}
\subsection{PDM Hamiltonian in physics}
\label{sub:physics} 
The motion of free carriers in materials (such as electrons and holes in semiconductors) of nonuniform chemical composition is sometimes described by means of a Hamiltonian with a position-dependent effective mass. It was soon realised that a mass varying in space comes with the problem of ordering
ambiguity in defining the kinetic  term of the Hamiltonian~\cite{Ross83}. In the physics literature, a PDM operator $H$ in divergence form as in~\eqref{eq:divergence} was considered as early as 1966 in the work of BenDaniel and Duke~\cite{BenDaniel66} on one-electron models in semiconductors. 
The setting of the present paper of a jump-discontinuous mass, appeared in the literature of quantum transport in semiconductors  as a model of abrupt junction of two materials~\cite{BenDaniel66,Dekar99,Pena17,Koc05}. In those papers, to extend the analogue of $H_{\operatorname{min}}$ in~\eqref{eq:Hmin} to a self-adjoint operator the {connection rules} of the wave function across the junction are the continuity of $\psi(x)$ and  $(1/m(x))\psi'(x)$. Together with the condition of vanishing wavefunction at the edge points, the connection rules identify the scale-free self-adjoint extension of $H_{\operatorname{min}}$ with boundary data connected by the unitary matrix
\begin{equation}
\label{eq:U_KirDir}
U=
\begin{pmatrix}
    1 &0 &0 &0\\
        0 &0 &-1 &0\\
            0 &-1 &0 &0\\
                0 &0 &0 &1
\end{pmatrix}.
\end{equation}
We stress that the choice~\eqref{eq:U_KirDir}  corresponds to just one of the infinitely many self-adjoint extensions $H_U$, $U\in\mathcal{U}(4)$.

\section{Spectral properties}
\label{sec:spectral_problem}
Consider now the spectral problem $H_U\psi_E=E\psi_E$ of the Hamiltonian $H_U$, $U\in\mathcal{U}(4)$, defined in Sec.~\ref{sub:model}. The next discussion follows by standard arguments adapted to our setting of different masses in $I_1$ and $I_2$.

\par

\subsection{Generalities}
\label{sec:spectral_gen} The eigenfunctions and eigenvalues of $H_U$ have the form 
\begin{equation}
\label{eq:eigenfunction}
\psi_E(x)=
\begin{cases}
c_1\rme^{\ii  k_1 x}+d_1 \rme^{-\ii  k_1 x} &\text{for $x\in I_1$}\\
c_2\rme^{\ii  k_2 x}+d_2 \rme^{-\ii  k_2 x} &\text{for $x\in I_2$}
\end{cases},\qquad 
E=\frac{\hbar^2k_1^2}{2m_1}=\frac{\hbar^2k_2^2}{2m_2},
\end{equation}
with wavenumbers $k_1,k_2$ related by
$\frac{k_1^2}{m_1}={\frac{k_2^2}{m_2}}$.
Imposing  $\psi_E\in D(H_U)$ from~\eqref{eq:Dom_H_U}, forces the vector of coefficients $(c_1,d_1,c_2,d_2)^T$ to be in the null space of a $4\times 4$ \emph{spectral matrix}
 \begin{equation}\label{eq:AU}
 	A_U(\kappa)=(\mathbb{I}+U)X(\kappa)+\kappa (\mathbb{I}-U)G Y(\kappa),
 \end{equation}
 with $G=\mathrm{diag}\left(\frac{1}{\sqrt{m_1}},\frac{1}{\sqrt{m_1}},\frac{1}{\sqrt{m_2}},\frac{1}{\sqrt{m_2}}\right)$,
 \begin{align}\label{eq:X}
 	X(\kappa)=
 	\begin{pmatrix}
 		\rme^{-\ii  \omega_1 \kappa } & \rme^{\ii  \omega_1 \kappa } & 0 & 0  & \\
 		1 & 1  & 0 & 0  &\\
		0 & 0 &  1 & 1 &\\
		0 & 0 &  \rme^{\ii  \omega_2 \kappa } & \rme^{-\ii  \omega_2 \kappa } &
        \end{pmatrix},
    \end{align}
 \begin{align}\label{eq:Y}
 	Y(\kappa)=
 	\begin{pmatrix}
	\rme^{-\ii  \omega_1 \kappa } & { -\rme^{\ii  \omega_1 \kappa }} & 0 & 0  & \\
	-1 & 1  & 0 & 0  &\\
	0 & 0 &  1 &  -1&\\
	0 & 0 &   -\rme^{\ii  \omega_2 \kappa } &  \rme^{-\ii  \omega_2 \kappa } &
	\end{pmatrix}
    \end{align}
 parametrised by the \emph{spectral parameter} $\kappa$ defined as
\begin{equation}
\label{eq:kappa}
\frac{k_1^2}{m_1}=\frac{k_2^2}{m_2}=:\kappa^2.
\end{equation}
The spectral matrix $A_{U}(\kappa)$ depends explicitly on the masses $m_i$,  and on the frequencies $\omega_i=\sqrt{m_i}\ell_i$, $i=1,2$. 
The requirement  $\psi_E\neq 0$ is therefore equivalent to the null space of $A_U(\kappa)$ being nontrivial. This occurs whenever $\kappa$ is a solution of the \emph{secular equation}
\begin{equation}
\label{eq:secular_varkappa}
S_U(\kappa)=0,
\end{equation}
where $S_U(\kappa)=\det A_U(\kappa)$ is the so-called \emph{spectral function}. 

If $\kappa$ is a zero of the spectral function $S_U$, a corresponding eigenfunction $\psi_{E_{\kappa}}$ with eigenvalue $E_{\kappa}=\hbar^2\kappa^2/2$ is given by~\eqref{eq:eigenfunction} with $(c_1,d_1,c_2,d_2)^T$ a nonzero vector in the (nontrivial) null space $\ker A_U(\kappa)$. 

\subsection{Scale-free b.c.\ and linear motions on the torus}
\label{sec:tori}
Recall that if $U=\mathbb{I}-2P$, with $P$ orthogonal projection, the equations for the boundary data do not mix the value of the functions $\gamma\psi$ and its derivatives $\nu\psi$ (scale-free b.c.). Recall (Proposition~\ref{prop:nonnegative}) that the eigenvalues of $H_{\mathbb{I}-2P}$ are nonnegative. In such a case, the spectral matrix has the form
\begin{equation}
    A_{\mathbb{I}-2P}(\kappa)=2\left[(\mathbb{I}-P)X(\kappa)+\kappa P G Y(\kappa)\right].
\end{equation}
We have the following result.
\begin{thm}
\label{thm:tori}
    Let $U=\mathbb{I}-2P\in\mathcal{U}(4)$, where $P$ is an orthogonal projection of rank $r=\operatorname{rank} P\in\{0,1,2,3,4\}$.
    The spectral function $S_{\mathbb{I}-2P}(\kappa)=\det A_{\mathbb{I}-2P}(\kappa)$ has the form 
\begin{equation}\label{eq:Spectral}
    S_{\mathbb{I}-2P}(\kappa)=\kappa^r   f_P(\omega_1\kappa,\omega_2\kappa) g_P(\kappa),
\end{equation}
where $g_P(\kappa) \neq 0$  for all $\kappa \in \R$ , and $f_P\colon \T^2\to \R$ is a trigonometric polynomial on the two-dimensional torus $\T^2=\R^2/(2\pi\Z)^2$ of side $2\pi$.
\end{thm}
\begin{proof}
Defining $Q=\mathbb{I}-P$, the matrix $QX(\kappa)+\kappa PGY(\kappa)$ can be rewritten in an eigenbasis of $P$ and $Q$, where $\kappa$ will appear as a multiplicative factor of the first $r$ rows, which yields $\det(A_U(\kappa))=\kappa^r\det(QX(\kappa)+PGY(\kappa))$. Moreover, $X(\kappa)$ and $Y(\kappa)$ depend on $\kappa$ only through the angles $(\varphi_1,\varphi_2)=(\omega_1 \kappa,\omega_2\kappa)$. By inspection of \eqref{eq:X} and \eqref{eq:Y} one  finds that $\det(QX(\kappa)+PGY(\kappa))$ is a trigonometric polynomial. After possibly factoring a non-vanishing term $g_P(\kappa)$, one gets \eqref{eq:Spectral}.
\end{proof}
In this paper, we will always assume the following.
\begin{assumption}[Nonresonant condition] 
\label{assump:1} The frequencies $\omega_1=\sqrt{m_1}\ell_1>0$ and $\omega_2=\sqrt{m_2}\ell_2>0$  are \emph{rationally independent} or \emph{nonresonant}, i.e. $\omega_1/\omega_2\notin\mathbb{Q}$.
\end{assumption}
We now recall a classical result on linear motions on the torus, see~\cite[Sec. 51]{Arnold89} 
\begin{prop}
\label{prop:mean}
    Consider, the linear flow $\varphi\colon \R\to\T^2$, with  $\varphi(\kappa)=\left(\varphi_1(\kappa),\varphi_2(\kappa)\right)$,
\begin{equation}
\label{eq:flow}
   \varphi_1(\kappa)=\omega_1\kappa,\quad \varphi_2(\kappa)=\omega_2\kappa\mod 2\pi.
\end{equation} 
Let $f:\mathbb{T}^2 \to \C$ be Riemann integrable. Then, under Assumption~\ref{assump:1}, the \emph{time average} of $f\circ\varphi$ exists and coincides with the \emph{space average} of $f$, i.e.
\begin{equation}
    \lim_{K\to +\infty}\frac{1}{K}\int_0^Kf\left(\varphi(\kappa)\right)d\kappa=\frac{1}{(2\pi)^2}\int_{\mathbb{T}^2}f(\varphi)d\varphi.
\end{equation}
In particular, the trajectory $\{\varphi(\kappa)\}_{\kappa\in \R}$ is dense in $\mathbb{T}^2$.
\end{prop}
By Theorem~\ref{thm:tori},  the condition for $E_{\kappa}=\hbar^2 \kappa^2/2$ to be an eigenvalue of $H_{\mathbb{I}-2P}$ can be written (excluding zero modes) as 
 \begin{equation}
     f_P(\omega_1\kappa,\omega_2\kappa)=0,
 \end{equation}
 where $f_P$ is a quasi-periodic function in $\mathbb{T}^2$ (a trigonometric polynomial). The equation $f_P(\varphi_1,\varphi_2)=0$ defines a curve $\Sigma_P$ embedded in $\mathbb{T}^2$. The zeros of $f_P(\omega_1\kappa,\omega_2\kappa)$ are given by the intersection of the linear flow~\eqref{eq:flow} with the curve $\Sigma_P$. See Fig.~\ref{fig:SF}. 
This observation, in the context of quantum graphs, is due to Barra and Gaspard~\cite{Barra2000}. An adaptation of their argument (see~\cite{Keating03} for an explicit construction) shows that there exists a probability measure on $\Sigma_P$ (known as \emph{Barra-Gaspard measure} in the literature of quantum graphs) that allows one to compute the Ces\`aro mean of general observables.
\begin{thm}
\label{thm:BG}
    Assume 
    that the flow~\eqref{eq:flow} is nowhere tangent to $\Sigma_P$. 
    Then, there exists a probability measure $\nu_P$ on $\Sigma_P$, such that
    \begin{equation}
        \lim_{E\to +\infty}\frac{1}{\#\{E_{\kappa}\leq E\}}\sum_{E_{\kappa}\leq E}g(E_{\kappa})=\int_{\Sigma_P}g(\varphi_1,\varphi_2)d\nu_P(\varphi_1,\varphi_2),
    \end{equation}
    for all piecewise continuous functions $g$.
\end{thm}
The explicit calculation of the measure $\nu_P$ for some scale free b.c. will be presented in  Sec.~\ref{sec:KirDir} below.  
\subsection{Uncoupled billiards}
\label{sec:uncoupled}
As a warm-up, consider the kinetic energy operator with Dirichlet b.c.\ $U=\mathbb{I}$, hence $r=0$ and $P=0$. The operator $H_{\mathbb{I}}$ is the direct sum of the kinetic energy (with mass $m_1$) in $I_1$ with Dirichlet b.c., and the kinetic energy (with mass $m_2$) in $I_2$ with Dirichlet b.c..  This corresponds to uncoupled billiards: 
\begin {center}
\begin {tikzpicture}[auto, node distance =1 cm and 3cm ,on grid ,
thick ,
state/.style ={ circle,inner sep=2pt ,top color =white , bottom color = gray ,
draw,black , text=blue , minimum width =.1 cm}]
\node[state] (C) {};
\node[state] (A) [ left=of C] {};
\node[state] (B) [right =of C] {};
\node[state] (D) [right =of B] {};
\node at (6 , -0.4) {$ \ell_2 $};
\node at (3 , -0.4) {$ 0^+ $};
\node at (0 , -0.4) {$ 0^- $};
\node at (-3 , -0.4) {$ -\ell_1$}; 
\draw[->-,red]  (A) to [bend right =0] (C);
\draw[->-,blue] (B) to [bend left =0] (D);
\end{tikzpicture}
\end{center}
The spectral matrix is a direct sum
\begin{align}
 	 A_{\mathbb{I}}(\kappa)=2
 	\begin{pmatrix}
 		\rme^{-\ii  \omega_1 \kappa } & \rme^{\ii  \omega_1 \kappa } & 0 & 0  & \\
 		1 & 1  & 0 & 0  &\\
		0 & 0 &  1 & 1 &\\
		0 & 0 &  \rme^{\ii  \omega_2 \kappa } & \rme^{-\ii  \omega_2 \kappa }   &
        \end{pmatrix},
    \end{align}
and the spectral function is
\begin{equation}
    S_{\mathbb{I}}(\kappa)=\det A_{\mathbb{I}}(\kappa)=-64\sin(\omega_1\kappa)\sin(\omega_2\kappa)=-64f_{0}(\omega_1 \kappa,\omega_2 \kappa).
\end{equation}
The function $f_{0}$ is factorised: 
\begin{equation}
f_{0}(\varphi_1,\varphi_2)=f_{0,1}(\varphi_1,\varphi_2)f_{0,2}(\varphi_1,\varphi_2)=\sin(\varphi_1)\sin(\varphi_2),
\end{equation} 
with $f_{0,1}(\varphi_1,\varphi_2)=   \sin(\varphi_1) $  and $f_{0,2}(\varphi_1,\varphi_2)=   \sin(\varphi_2) $. Moreover $\Sigma_0=\Sigma_{0,1}\cup \Sigma_{0,2}$, where
\begin{equation}
\Sigma_{0,1}=\left\{ (\varphi_1, \varphi_2) \in \T^2: f_{0,1}(\varphi_1,\varphi_2)=0\right\}, \quad \Sigma_{0,2}=\left\{ (\varphi_1, \varphi_2) \in \T^2: f_{0,2}(\varphi_1,\varphi_2)=0\right\}.
\end{equation}
The spectrum of $H_{\mathbb{I}}$ is 
\begin{eqnarray}
\sigma\left( H_{\mathbb{I}}\right)&=&  \bigcup_{j=1}^2 \left \{\frac{\hbar^2}{2}\kappa^2 : \kappa\in \mathcal{K}_j  \right\},
\end{eqnarray}
where for all $j=1,2$:
\begin{equation}\label{eq:sigma12}
\mathcal{K}_j = \left\{\kappa \neq 0 : (\omega_1\kappa,\omega_2\kappa)\mod 2\pi \in \Sigma_{j,0}\right\}=\left\{\frac{\pi}{\omega_j}n \colon n \in \mathbb{Z} \setminus \{0\}\right\}.
\end{equation}
By Assumption~\ref{assump:1}, the sets $\mathcal{K}_1$ and $\mathcal{K}_2$ are disjoint.

If $\kappa=(\pi/\omega_1)n_1\in\mathcal{K}_1$, then,
\begin{equation}
    \operatorname{ker}A_{\mathbb{I}}\left(\frac{\pi }{\omega_1} n_1\right)=   \operatorname{ker}    
    \begin{pmatrix}
 		\rme^{-\ii  \pi n_1 } & \rme^{\ii  \pi n_1 } & 0 & 0  & \\
 		1 & 1  & 0 & 0  &\\
		0 & 0 &  1 & 1 &\\
		0 & 0 & \rme^{\ii \pi \frac{\omega_2}{\omega_1}n_1} &\rme^{-\ii \pi \frac{\omega_2}{\omega_1}n_1}  &
        \end{pmatrix}=\operatorname{span}\left\{\left(\begin{array}{c}1 \\-1\\0\\0 \end{array}\right)\right\}.
\end{equation}
If $\kappa=(\pi/\omega_2)n_2\in \mathcal{K}_2$, then,
\begin{equation}
    \operatorname{ker}A_{\mathbb{I}}\left(\frac{\pi }{\omega_2} n_2\right)=   \operatorname{ker}    
    \begin{pmatrix}
 		\rme^{ -\ii\pi \frac{\omega_1}{\omega_2}n_2} &\rme^{\ii\pi \frac{\omega_1}{\omega_2}n_2}& 0 & 0  & \\
 		1 & 1  & 0 & 0  &\\
		0 & 0 &  1 & 1 &\\
		0 & 0 &\rme^{ \ii\pi n_2} &\rme^{-\ii\pi n_2}  &
        \end{pmatrix}=\operatorname{span}\left\{\left(\begin{array}{c}0\\0 \\1 \\-1 \end{array}\right)\right\}.
\end{equation}

We conclude that the normalised eigenfunctions of $H_{\mathbb{I}}$ are
\begin{equation}
    \psi_{E_{\kappa}}(x)=
    \begin{cases}
 \displaystyle\sqrt{\frac{2}{\ell_1}}\sin\left(\frac{n_1 \pi}{\ell_1}x\right)\chi_{I_1}(x) &\text{if $\kappa=\frac{\pi}{\omega_1}n_1\in \mathcal{K}_1$},\\
  \displaystyle\sqrt{\frac{2}{\ell_2}}\sin\left(\frac{n_2 \pi}{\ell_2}x\right)\chi_{I_2}(x) &\text{if $\kappa=\frac{\pi}{\omega_2}n_2\in \mathcal{K}_2$}.
    \end{cases}
\end{equation}
The eigenstates with energy $\frac{\hbar^2 \pi^2 n_1^2}{2m_1\ell_1^2}$ are localised in $I_1$ and have leaning $\mathcal{L}(\psi_{E_{\kappa}})=-1$; those with energy $\frac{\hbar^2 \pi^2 n_2^2}{2m_2\ell_2^2}$ are localised in $I_2$ and have leaning $\mathcal{L}(\psi_{E_{\kappa}})=1$. 

The energy levels counting function satisfies the Weyl law
\begin{align}
    N(E)&=\#\{E_{\kappa}\leq E\}\\
    &=\#\left\{n_1\in\N\setminus \{0\}\colon \frac{\hbar^2 \pi^2 n_1^2}{2\omega_1^2}\leq E\right\}+\#\left\{n_2\in\N \setminus \{0\}\colon\frac{\hbar^2 \pi^2 n_2^2}{2\omega_2^2}\leq E\right\}\\
    &\sim\frac{\omega_1+\omega_2}{\hbar \pi}\sqrt{2E},\quad \text{as $E\to +\infty$}.
\end{align}
It follows that for $U=\mathbb{I}$, the sequence $\left\{\mathcal{L}(\psi_{E_{\kappa}})\right\}_{E_{\kappa}}$ contains only two subsequences taking values $\pm 1$, and
\begin{align}
\label{eq:upper_limit_D}
\limsup_{\kappa \to +\infty} \mathcal{L}(\psi_{E_{\kappa}})&=1,\\
\label{eq:Cesaro_limit_D}
    \lim_{E\to +\infty}\frac{1}{\#\{E_{\kappa}\leq E\}}\sum_{E_{\kappa}\leq E}\mathcal{L}(\psi_{E_{\kappa}})&=\frac{\omega_2-\omega_1}{\omega_2+\omega_1}=\frac{\sqrt{m_2}\ell_2-\sqrt{m_1}\ell_1}{\sqrt{m_2}\ell_2+\sqrt{m_1}\ell_1},\\
    \label{eq:lower_limit_D}
\liminf_{\kappa \to +\infty} \mathcal{L}(\psi_{E_{\kappa}})&=-1.
\end{align}
Mutatis mutandi, the same analysis can be repeated for Neumann b.c.\ ($U=-\mathbb{I}$, hence $r=4$ and $P=\mathbb{I}$). Imposing the vanishing of $\nu\psi$, we find the spectral matrix 
\begin{equation}
    A_{-\mathbb{I}}(\kappa)=2\kappa
   \begin{pmatrix}
 		\frac{\rme^{-\ii  \omega_1 \kappa }}{\sqrt{m_1}} & \frac{ -\rme^{\ii  \omega_1 \kappa }}{\sqrt{m_1}} & 0 & 0  & \\
 		\frac{-1}{\sqrt{m_1}} & \frac{1}{\sqrt{m_1}}  & 0 & 0  &\\
		0 & 0 &  \frac{1}{\sqrt{m_2}} &  \frac{-1}{\sqrt{m_2}} &\\
		0 & 0 &   \frac{-\rme^{\ii  \omega_2 \kappa }}{\sqrt{m_2}} &  \frac{\rme^{-\ii  \omega_2 \kappa }}{\sqrt{m_2}} &
        \end{pmatrix},
\end{equation}
and the spectral function is again factorised
\begin{equation}
    S_{-\mathbb{I}}(\kappa)=\det A_{-\mathbb{I}}(\kappa)=-\frac{64\kappa^4}{m_1m_2}\sin(\omega_1\kappa)\sin(\omega_2\kappa)=-\frac{64\kappa^4}{m_1m_2}f_{\mathbb{I}}(\varphi_1,\varphi_2).
\end{equation}
In this case $f_{\mathbb{I}}(\varphi_1,\varphi_2)=f_{\mathbb{I},1}(\varphi_1,\varphi_2)f_{\mathbb{I},2}(\varphi_1,\varphi_2)$, with $f_{\mathbb{I},1}(\varphi_1,\varphi_2)=\sin(\varphi_1)$ and\\  $~{f_{\mathbb{I},2}(\varphi_1,\varphi_2)=\sin(\varphi_2)}$.
Moreover, $\Sigma_{\mathbb{I}}=\Sigma_{\mathbb{I},1}\cup \Sigma_{\mathbb{I},2}$ with
\begin{equation}
\Sigma_{\mathbb{I},1}=\left\{ (\varphi_1, \varphi_2) \in \T^2: f_{\mathbb{I},1}(\varphi_1,\varphi_2)=0\right\}, \quad \Sigma_{\mathbb{I},2}=\left\{ (\varphi_1, \varphi_2) \in \T^2: f_{\mathbb{I},2}(\varphi_1,\varphi_2)=0\right\}.
\end{equation}
The spectrum of $H_{\mathbb{I}}$ is 
\begin{eqnarray}
\sigma\left( H_{-\mathbb{I}}\right)&=&  \bigcup_{j=0}^2 \left \{\frac{\hbar^2}{2}\kappa^2 : \kappa\in \mathcal{K}_j  \right\},
                                            \end{eqnarray}
where $\mathcal{K}_0= \{0\}$ and for $j=1,2$
\begin{equation}
\mathcal{K}_j = \left\{\kappa \neq 0 : (\omega_1\kappa,\omega_2\kappa) \mod 2\pi \in \Sigma_{j,\mathbb{I}}\right\}=\left\{\frac{\pi}{\omega_j}n \colon n \in \mathbb{Z} \setminus \{0\}\right\}.
\end{equation}
By Assumption~\ref{assump:1}, the sets $\mathcal{K}_1$ and $\mathcal{K}_2$ are disjoint.
 
The normalised eigenfunctions of $H_{-\mathbb{I}}$ for positive $E_{\kappa}>0$, are
\begin{equation}
    \psi_{E_{\kappa}}(x)=
    \begin{cases}
 \displaystyle\sqrt{\frac{2}{\ell_1}}\cos\left(\frac{n_1 \pi}{\ell_1}x\right)\chi_{I_1}(x) &\text{if $\kappa=\frac{\pi}{\omega_1}n_1\in \mathcal{K}_1$},\\
  \displaystyle\sqrt{\frac{2}{\ell_2}}\cos\left(\frac{n_2 \pi}{\ell_2}x\right)\chi_{I_2}(x) &\text{if $\kappa=\frac{\pi}{\omega_2}n_2\in \mathcal{K}_2$},
    \end{cases}
\end{equation}
while the normalised eigenfunctions for $E_{\kappa}=0$ (zero modes) are
\begin{equation}\label{eq:zeromod}
    \psi_{0}(x)=
 \displaystyle\sqrt{\frac{1}{\ell_1}}\alpha\chi_{I_1}(x)+\sqrt{\frac{1}{\ell_2}}\beta\chi_{I_2}(x) ,\quad |\alpha|^2+|\beta|^2=1.
\end{equation}
Therefore, excluding the zero modes, the eigenfunctions of $H_{-\mathbb{I}}$ behave as those of $H_{\mathbb{I}}$. In particular, \eqref{eq:upper_limit_D}-\eqref{eq:Cesaro_limit_D}-\eqref{eq:lower_limit_D} hold true.
\par

General `uncoupled billiards' are those self-adjoint extensions $H_U$ with 
$$
U=\begin{pmatrix} u_{11}& u_{12} &0 &0 &\\
u_{21}& u_{22} &0 &0 &\\
0 &0 & u_{33}& u_{34}\\
0 &0 & u_{43}& u_{44}\\
\end{pmatrix} \in \mathcal{U}(4).
$$
The above examples ($U=\mathbb{I}$ and $U=-\mathbb{I}$) suggest that~\eqref{eq:upper_limit_D}-\eqref{eq:Cesaro_limit_D}-\eqref{eq:lower_limit_D} extend to generic uncoupled b.c.\ under the nonresonant condition.

\section{Spectral properties for scale-free boundary conditions}
\label{sec:scalefree}
In this section we derive the explicit solutions of the spectral problem of a two-edges quantum graph (edges $I_1$ and $I_2$ with mass $m_1$, and $m_2$) for the special class of scale-free b.c.\ defined by:
\begin{enumerate}
    \item the continuity conditions of $\psi$ and $(1/m)\psi'$   (zero Kirchhoff b.c.) at the connected vertices, namely vertices of degree at least two;
    \item Dirichlet b.c. at the pendant vertices, namely vertices of degree one (a similar analysis can be performed with Neumann b.c. at the pendant vertices).
\end{enumerate}
In Sec.~\ref{sec:uncoupled} we have considered the case of Dirichlet or Neumann b.c.\ at all vertices, corresponding to uncoupled (non connected) edges. 

Before analyzing the various cases in detail, it is necessary to introduce some useful notation. Recall that any unitary matrix $U \in \mathcal{U}(4)$ encoding scale-free b.c.\ is of the form $U=\mathbb{I}-2P$, where $P$ is a $r$-dimensional orthogonal projection, with $r \in \{0,1,2,3,4\}$.  
Recall that the spectral function has the form $S_{\mathbb{I}-2P}(\kappa)=\kappa^r f_P(\omega_1\kappa,\omega_2\kappa)g_P(\omega_1\kappa,\omega_2\kappa)$, with $f_P:\T^2 \to \mathbb{R}$ a trigonometric polynomial. In some of the examples considered in the following subsections, $f_P$  is written in a more convenient way (by means of the tangent half-angle formulae $\cos \varphi_i= (1-t_i^2)/(1+t_i^2)$, $\sin \varphi_i= 2t_i(1+t_i^2)$, with $t_i=\tan(\varphi_i/2)$, $i=1,2$), as
\begin{equation}
f_P(\varphi_1, \varphi_2)=\prod_{j=1}^{q}f_{P,j}(\varphi_1, \varphi_2),
\end{equation}
for some $q \in \{1,2,3,4\}$, with $f_{P,j}: \T^2 \to \mathbb{R}$ for all $j=1, \dots, q$. Hence 
\begin{equation}
\Sigma_P= \left\{ (\varphi_1, \varphi_2) \in \T^2: f_{P}(\varphi_1,\varphi_2)=0\right\} = \bigcup_{j=1}^{q} \Sigma_{P,j},
\end{equation}
where $\Sigma_{P,j}= \left\{ (\varphi_1, \varphi_2) \in \T^2: f_{P,j}(\varphi_1,\varphi_2)=0\right\}$, for all $j=1, \dots,q$.
The spectrum of $H_{\mathbb{I}-2P}$ is 
\begin{eqnarray}
\sigma\left( H_{\mathbb{I}-2P} \right)&=&  \bigcup_{j=0}^q \left \{E_{\kappa}=\frac{\hbar^2}{2}\kappa^2 : \kappa\in \mathcal{K}_j  \right\},
\end{eqnarray}
where for $j=1,\dots,q$
\begin{equation}
\mathcal{K}_j = \left\{\kappa \neq 0 : (\omega_1\kappa,\omega_2\kappa) \mod 2\pi \in \Sigma_{P,j}\right\},
\end{equation}
while  $\mathcal{K}_0= \{0\}$ if $0$ is an eigenvalue of $H_{\mathbb{I}-2P}$  and $\mathcal{K}_0=\emptyset$ otherwise. By the nonresonant condition of $\omega_1,\omega_2$, the sets $\mathcal{K}_j$'s  are disjoint.

For $E_\kappa >0$ the eigenfunctions $\psi_{E_{\kappa}}$ are in the form~\eqref{eq:eigenfunction} with $(c_1,d_1,c_2,d_2)^T$ non-zero element of $\operatorname{ker}A_{\mathbb{I}-2P}(\kappa)$ and the corresponding leaning is
\begin{equation}\label{eq:leaningintro}
    \mathcal{L}(\psi_{E_{\kappa}})=\frac{\int_{I_2}|\psi_{E_{\kappa}}(x)|^2\, dx-\int_{I_1}|\psi_{E_{\kappa}}(x)|^2\, dx}{\int_{I_2}|\psi_{E_{\kappa}}(x)|^2\, dx+\int_{I_1}|\psi_{E_{\kappa}}(x)|^2\, dx},
\end{equation}
where
\begin{align}
\int_{I_1}|\psi_{E_{\kappa}}(x)|^2\, dx&=\ell_{1}\left[|c_{1}|^2+|d_{1}|^2+2\operatorname{Re}\left(c_{1}\overline{d_{1}}\rme^{-\ii \omega_1 \kappa}\frac{\sin(\omega_1 \kappa)}{\omega_1 \kappa}\right)\right],\\
\int_{I_2}|\psi_{E_{\kappa}}(x)|^2\, dx&=\ell_{2}\left[|c_{2}|^2+|d_{2}|^2+2\operatorname{Re}\left(c_{2}\overline{d_{2}}\,\rme^{\ii \omega_2 \kappa}\,\frac{\sin(\omega_2 \kappa)}{\omega_2 \kappa}\right)\right].
\end{align}
As we will show in the examples, the coefficients $c_1,d_1$ depend only on $\omega_1\kappa$, while $c_2,d_2$ depend only on $\omega_2\kappa$, thus the leaning  $\mathcal{L}(\psi_{E_\kappa})$ in \eqref{eq:leaningintro} depends on $\omega_1\kappa$ and $\omega_2\kappa$. The behaviour of $\mathcal{L}(\psi_{E_\kappa})$ is erratic in $\kappa$ and strongly depends on the b.c. $U$, see Fig.~\ref{fig:LeaningAll}. On the contrary, the Ces\`aro mean of the leaning assumes the same value
\begin{equation}\label{cesmeanlimintro}
       \lim_{E\to +\infty}\frac{\sum_{E_{\kappa}\leq E} \mathcal{L}(\psi_{E_{\kappa}})}{\#\{E_{\kappa}\leq E\}}=\frac{\sqrt{m_2}\ell_2-\sqrt{m_1}\ell_1}{\sqrt{m_2}\ell_2+\sqrt{m_1}\ell_1}.
\end{equation}
To prove this we will see that for all $j=1,\dots,q$, the angle $\theta_j$ between the normal to the set $\Sigma_{P,j}$ and the linear flow on the torus \eqref{eq:flow} is such that \begin{equation}
    \cos\theta_j=\frac{\omega\cdot \nabla f_{P,j}(\varphi_1,\varphi_2)}{\|\omega\|\|\nabla f_{P,j}(\varphi_1,\varphi_2)\|} \neq0, \quad (\varphi_1,\varphi_2)\in\Sigma_{P,j},
\end{equation}
where $\omega=(\omega_1,\omega_2)$ is the vector of frequencies and $\nabla f_{P,j}$ is the gradient of $f_{P,j}$. Hence, by Theorem~\ref{thm:BG} we can compute the Barra-Gaspard measure $\nu_P$ supported on the $q$ components of $\Sigma_P$ (the corresponding restrictions of $\nu_P$ are denoted with $\nu_{P,1}, \dots, \nu_{P,q}$). See Fig.~\ref{fig:BG}. Then, the Ces\`aro mean of the leaning can be computed as a space average on $\Sigma_P$ with respect to the measure $\nu_P$.

\begin{figure}
    \centering
\raisebox{4.5mm}{
\begin{tikzpicture}[scale=0.95, transform shape,auto, node distance =1 cm and 3cm ,on grid ,
thick ,
state/.style ={ circle,inner sep=2pt ,top color =white , bottom color = gray ,
draw,black , text=blue , minimum width =.1 cm}]
\node[state] (C) {};
\node[state] (A) [ left=of C] {};
\node[state] (B) [right =of C] {};
\node at (3 , -0.4) {$ \ell_2 $};
\node at (0.25 , -0.4) {$ 0^+ $};
\node at (-0.25 , -0.4) {$ 0^-$}; 
\node at (-3 , -0.4) {$ -\ell_1$}; 
\draw[->-,red]  (A) to [bend right =0] (C);
\draw[->-,blue]  (C) to [bend left =0]  (B);
\end{tikzpicture}}\qquad 
\begin {tikzpicture}[scale=0.95, transform shape,auto, node distance =1 cm and 3cm ,on grid ,
thick ,
state/.style ={ circle,inner sep=2pt ,top color =white , bottom color = gray ,
draw,black , text=blue , minimum width =.1 cm}]
\node[state] (C) {};
\node[state] (A) [ left=of C] {};
\node at (-3.2 , 0.4) {$ \ell_2 $};
\node at (0.2 , 0.4) {$ 0^+ $};
\node at (0.2 , -0.4) {$ 0^- $};
\node at (-3.2 , -0.4) {$ -\ell_1$}; 
\draw[->-,red]   (A) to [bend right =70]  (C);
\draw[->-,blue]  (C) to [bend right =70]  (A);
\end{tikzpicture}\\
\hspace{20pt} (a) The segment \hspace{100pt} (b) The ring\\
\qquad\qquad
\begin {tikzpicture}[auto, node distance =1 cm and 3cm ,on grid ,
thick ,
state/.style ={ circle,inner sep=2pt ,top color =white , bottom color = gray ,
draw,black , text=blue , minimum width =.1 cm}]
\node[state] (C) {};
\node[state] (A) [ left=of C] {};
\node at (-0.1 , 0.4) {$ \ell_2 $};
\node at (0.5 , -0) {$ 0^+ $};
\node at (-0.1 , -0.4) {$ 0^-$}; 
\node at (-3 , -0.4) {$ -\ell_1$}; 
\draw[->-,red]  (A) to [bend right =0]  (C);
\draw[->-,blue,every loop/.style={looseness=50}] (C)
         to  [in=45,out=-45,loop] ();
\end{tikzpicture} 
\begin{tikzpicture}[auto, node distance =1 cm and 3cm ,on grid ,
thick ,
state/.style ={ circle,inner sep=2pt ,top color =white , bottom color = gray ,
draw,black , text=blue , minimum width =.1 cm}]
\node[state] (C) {};
\node at (-0 , 0.5) {$ \ell_2 $};
\node at (0.5 , -0) {$ 0^+ $};
\node at (0.1 , -0.5) {$ 0^-$}; 
\node at (-0.6 , -0) {$ -\ell_1$}; 
\draw[->-,red,every loop/.style={looseness=50}] (C)
         to  [in=-135,out=135,loop] ();
\draw[->-,blue,every loop/.style={looseness=50}] (C)
         to  [in=45,out=-45,loop] ();
\end{tikzpicture}\\
\vspace{-35pt}
\hspace{20pt} (c) The pendant \hspace{100pt} (d) The rose\\
    \caption{Pictorial representation of the boundary conditions investigated: (a) the segment, equation~\eqref{eq:Segment}, (b) the ring, equation~\eqref{eq:Ring}, (c) the pendant, equation~\eqref{eq:Pendant}, and (d) the rose, equation~\eqref{eq:bc_rose}}
    \label{fig:pictorialrepresentation}
\end{figure}
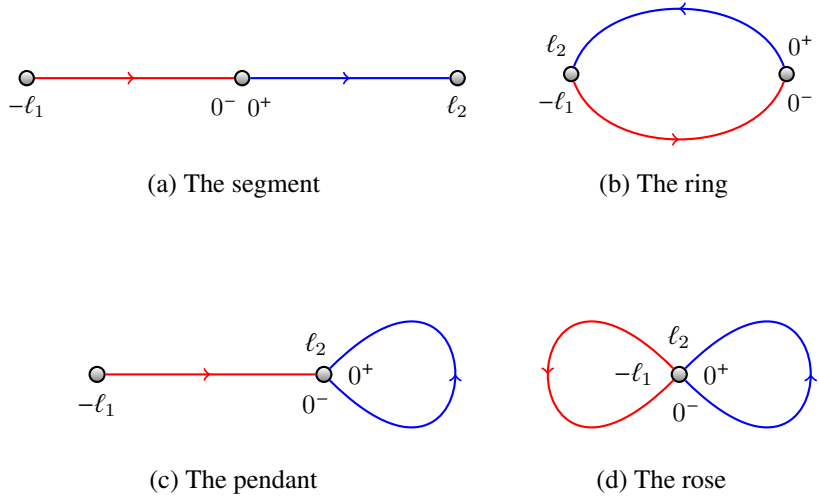
 \begin{figure}[th]
  \centering
  \begin{subfigure}[b]{0.42\textwidth}
    \includegraphics[width=\textwidth]{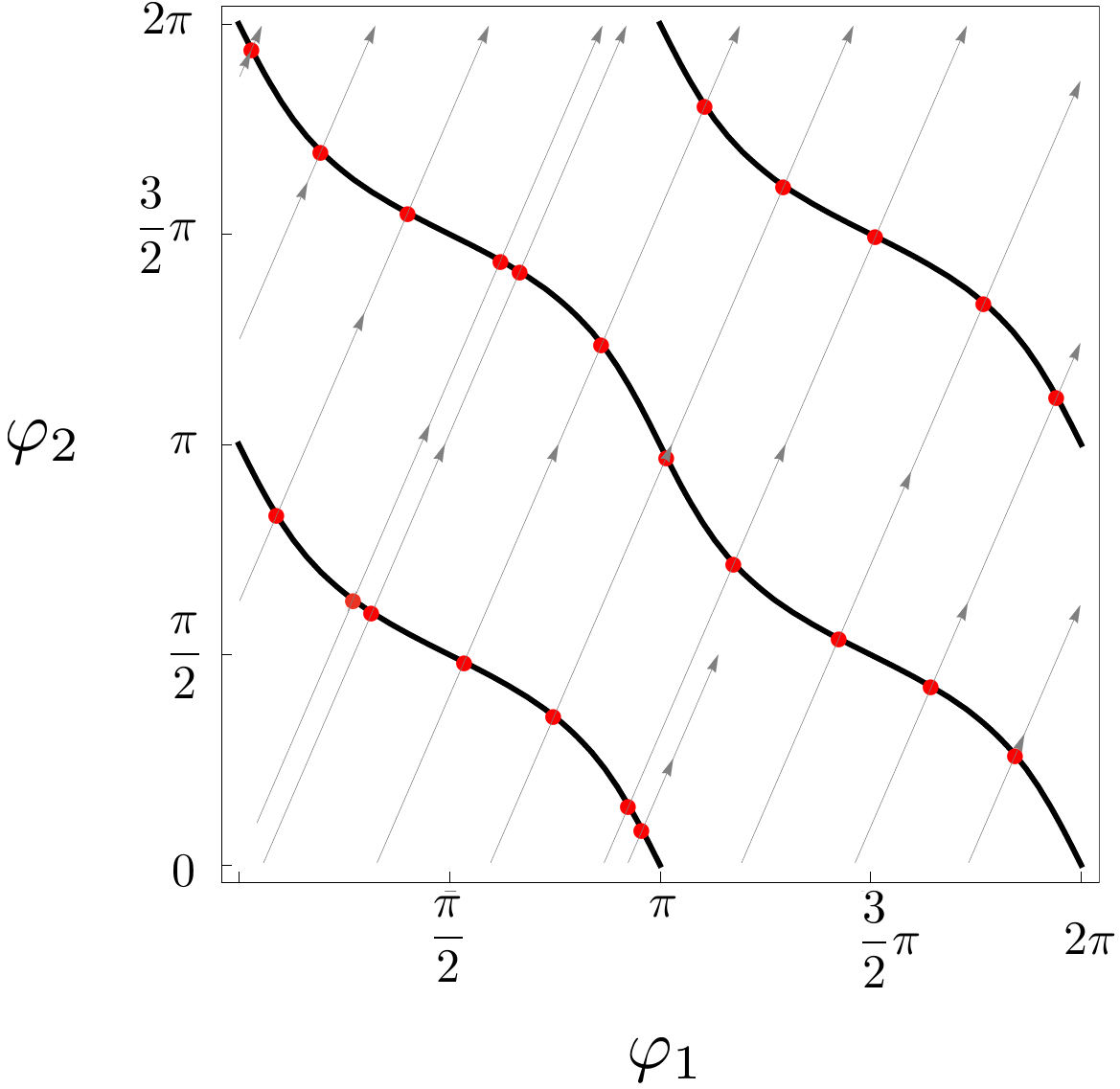} 
    \caption{The segment, Eq.~\eqref{eq:SF_KirDir}}
    \label{fig:subSF1}
  \end{subfigure}
  \hfill
  \begin{subfigure}[b]{0.42\textwidth}
    \includegraphics[width=\textwidth]{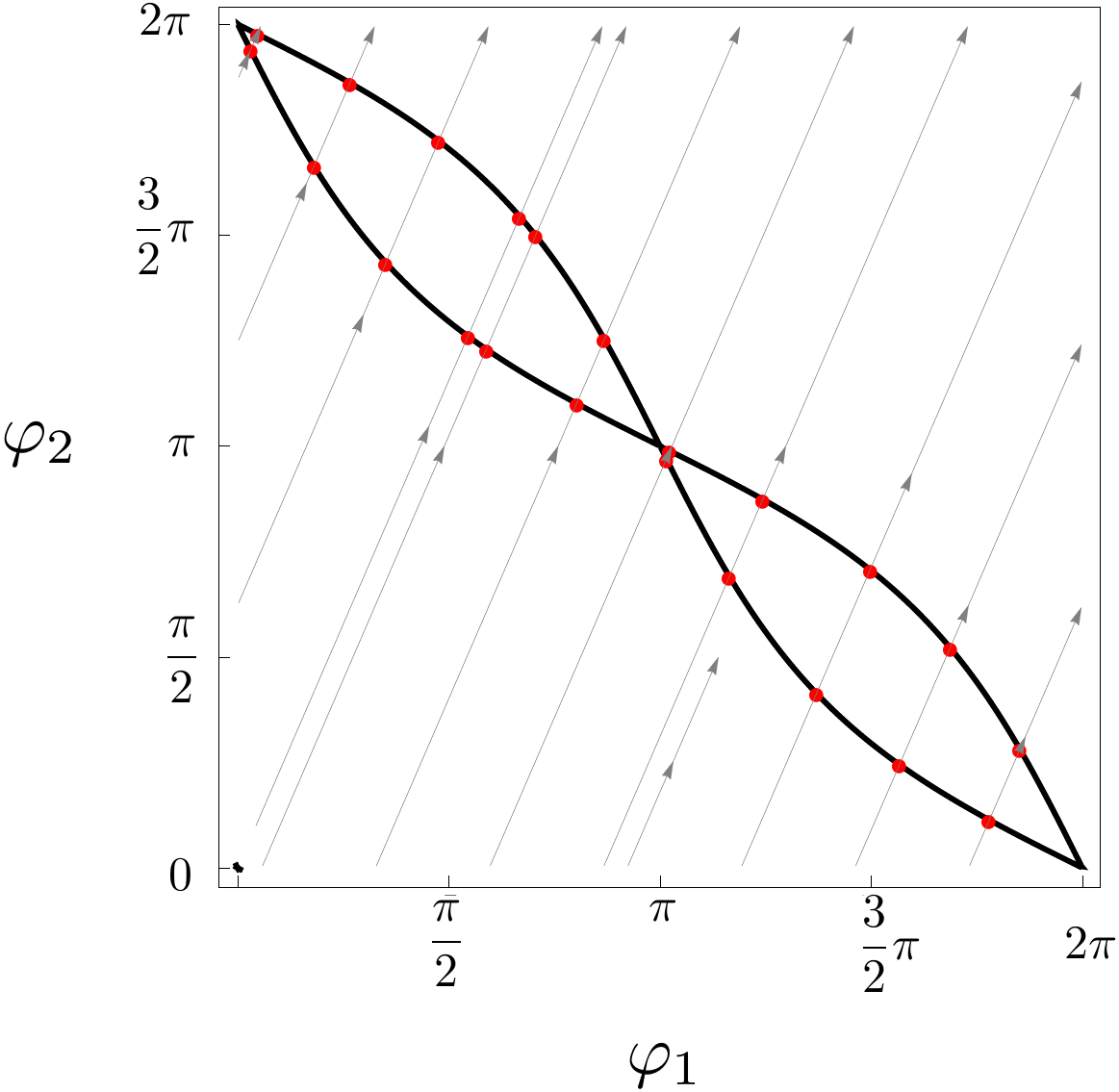}
    \caption{The ring, Eq.~\eqref{eq:SF_KirRing}.}
    \label{fig:subSF2}
  \end{subfigure}\vspace{.5cm}
    \hfill
  \begin{subfigure}[b]{0.42\textwidth}
    \includegraphics[width=\textwidth]{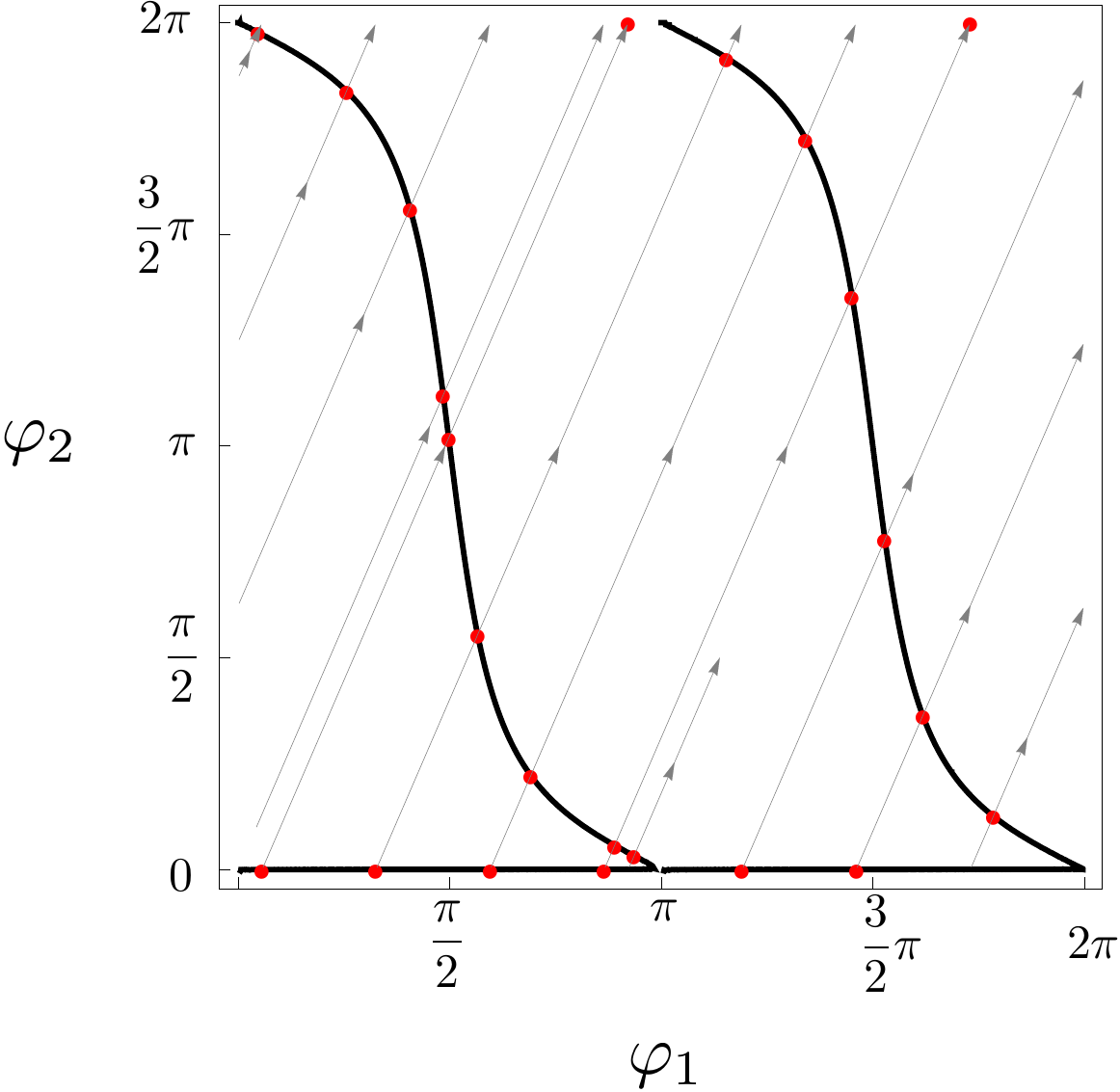}
    \caption{The single pendant, Eq.~\eqref{eq:SF_KirPend}.}
    \label{fig:subSF3}
  \end{subfigure}
  \hfill
  \begin{subfigure}[b]{0.42\textwidth}
    \includegraphics[width=\textwidth]{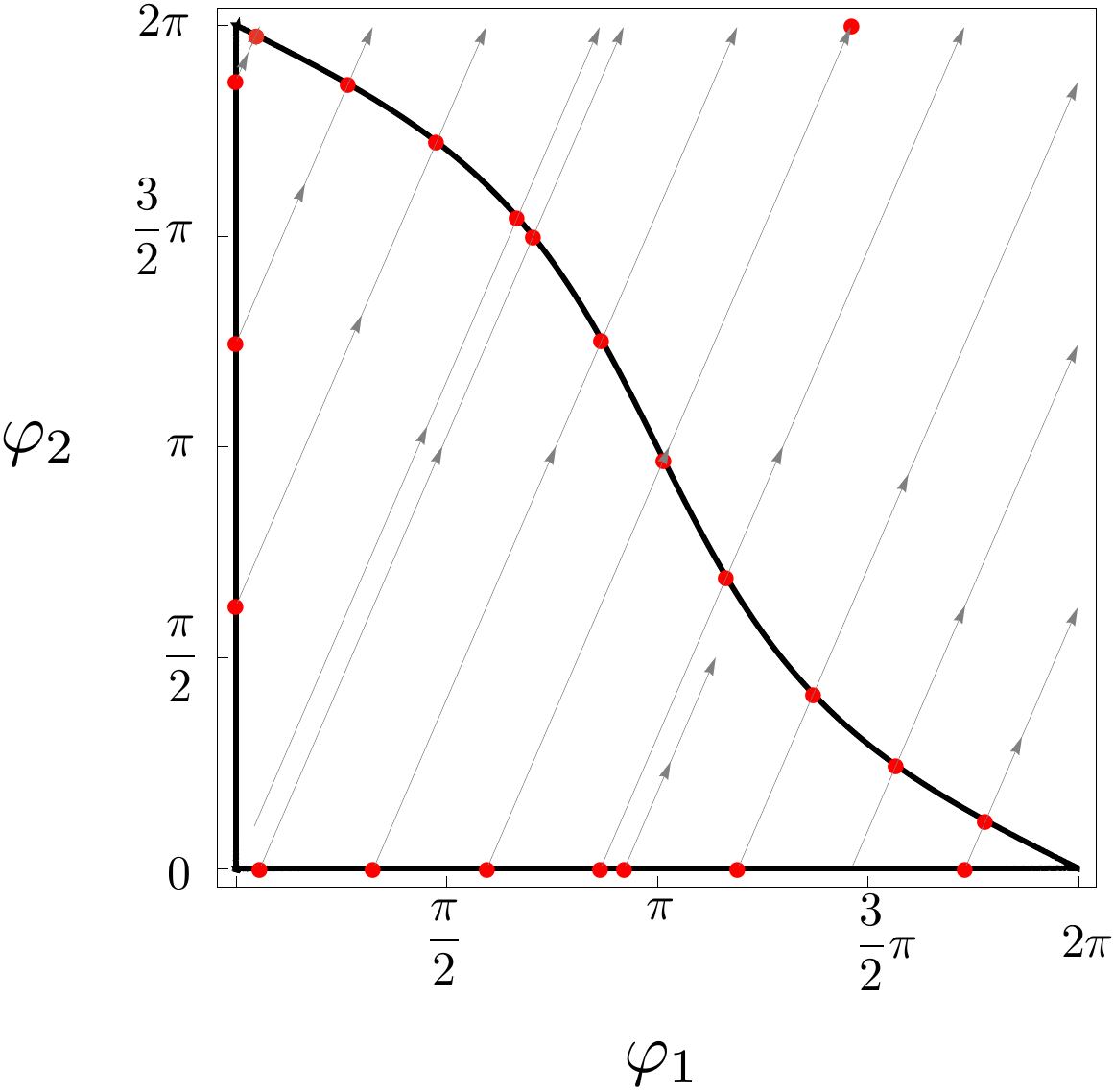}
    \caption{The rose, Eq.~\eqref{eq:SF_KirRose}.}
    \label{fig:subSF4}
  \end{subfigure}
    \caption{Zero set $\Sigma_P$ of the spectral functions  on the torus $\mathbb{T}^2$, along with the linear flow~\eqref{eq:flow}. Here $m_1=16$, $m_2=1$, $\ell_1=\rme$, and $\ell_2=2\pi$. The corresponding frequencies  $\omega_1=4\rme$ and $\omega_2=\pi$ are nonresonant.}
    \label{fig:SF} 
\end{figure}

\subsection{Zero Kirchhoff b.c.\ at the junction}
\label{sec:KirDir}
Two segments with Kirchhoff b.c.\ at the junction in $x=0$ and Dirichlet at the pendant vertices (see Figure~\ref{fig:pictorialrepresentation}):
\begin{equation}\label{eq:Segment}
    \begin{cases}
    \psi(0^-)=\psi(0^+),\\\\
    \psi(-\ell_1)=0,\\\\
    \psi(\ell_2)=0, \\\\
       \dfrac{1}{m_1}\psi'(0^-)- \dfrac{1}{m_2}\psi'(0^+)=0.
    \end{cases}
\end{equation}

These b.c.\ correspond to a projection $P=\ketbra{u}{u}$ (in Dirac notation) of rank $r=1$ with $u=\frac{1}{\sqrt{2}}(0,1,1,0)^T$, the corresponding matrix $U=\mathbb{I}-2P$ is in~\eqref{eq:U_KirDir}, and  
\begin{align}
\label{eq:SF_KirDir}
f_P(\varphi_1,\varphi_2)=\sqrt{m_1}\sin \varphi_1   \cos \varphi_2  +\sqrt{m_2}  \sin  \varphi_2 \cos  \varphi_1.
\end{align}
Since $0$ is not an eigenvalue of $H_{\mathbb{I}-2P}$, the spectrum of $H_{\mathbb{I}-2P}$ is
\begin{equation}
\sigma(H_{\mathbb{I}-2P})=\left\{ E_{\kappa}= \frac{\hbar^2}{2} \kappa^2: \kappa \in \mathcal{K}\right\},
\end{equation}
where
\begin{equation}
\mathcal{K}=\left\{ \kappa \neq 0 : f_P(\omega_1\kappa,\omega_2\kappa)=0 \right\}.
\end{equation}
For all $\kappa \in \mathcal{K}$ the null space of the spectral matrix $A_{\mathbb{I}-2P}(\kappa)$ is nontrivial and the non zero elements  in $\operatorname{ker}A_{\mathbb{I}-2P}(\kappa)$ are 
\begin{equation}
\label{eq:coeff_KirDir}
\left(c_1,d_1,c_2,d_2\right)^T=
\alpha \left(\frac{\rme^{\ii \omega_1 \kappa}}{\sin(\omega_1 \kappa)},-\frac{\rme^{-\ii\omega_1 \kappa}}{\sin(\omega_1 \kappa)},-\frac{\rme^{-\ii \omega_2 \kappa}}{\sin(\omega_2 \kappa)},\frac{\rme^{\ii \omega_2 \kappa}}{\sin(\omega_2 \kappa)}\right)^T,
\end{equation}
with $\alpha \in \C$.
Inserting this into~\eqref{eq:eigenfunction}, we get the eigenfunctions 
\begin{equation}
\psi_{E_{\kappa}}(x)=\frac{ \sin \left(\kappa \sqrt{m_1} \left( x+\ell_1\right)\right)}{\sin \left(\kappa  \omega_1\right)}\chi_{I_1}(x)-\frac{ \sin \left(\kappa  \sqrt{m_2} \left(x-\ell_2\right)\right)}{\sin \left(\kappa  \omega_2\right)}\chi_{I_2}(x).
\end{equation}
We can compute the corresponding leaning and, using that $f_P(\omega_1 \kappa, \omega_2 \kappa)=0$, we get the following asymptotic expansion
\begin{equation}
\label{eq:leaning_DirKir}
     \mathcal{L}(\psi_{E_{\kappa}})=\frac{\left(m_1\sin^2(\omega_1 \kappa)+m_2\cos^2(\omega_1 \kappa)\right)\ell_2-m_1\ell_1}{\left(m_1\sin^2(\omega_1 \kappa)+m_2\cos^2(\omega_1\kappa)\right)\ell_2+m_1\ell_1}+O\left(\frac{1}{\kappa}\right),\quad \text{as $\kappa\to +\infty$}.
\end{equation}
By Proposition~\ref{prop:mean} and Theorem~\ref{thm:BG}, the sequence of values of the leaning $\{ \mathcal{L}(\psi_{E_{\kappa}})\}_{E_{\kappa}}$ densely fills the range of the function on the right-hand side of ~\eqref{eq:leaning_DirKir}. In particular, we have
\begin{align}
    \limsup_{\kappa \to + \infty} \mathcal{L}(\psi_{E_{\kappa}})&=\sup_{\varphi \in[0,2\pi]}\frac{\left(m_1\sin^2\varphi+m_2\cos^2\varphi\right)\ell_2-m_1\ell_1}{\left(m_1\sin^2\varphi+m_2\cos^2\varphi\right)\ell_2+m_1\ell_1},\\
\liminf_{\kappa \to + \infty}  \mathcal{L}(\psi_{E_{\kappa}})&=\inf_{\varphi\in[0,2\pi]}
\frac{\left(m_1\sin^2\varphi+m_2\cos^2\varphi\right)\ell_2-m_1\ell_1}{\left(m_1\sin^2\varphi+m_2\cos^2\varphi\right)\ell_2+m_1\ell_1},
\end{align}
and hence we obtain 
\begin{align}
\label{eq:suplim}
\limsup_{\kappa \to + \infty} \mathcal{L}(\psi_{E_{\kappa}})&=\frac{\ell_2-\ell_1}{\ell_2+\ell_1},\\
\label{eq:inflim}
\liminf_{\kappa \to + \infty} \mathcal{L}(\psi_{E_{\kappa}})&=\frac{m_2\ell_2-m_1\ell_1}{m_2\ell_2+m_1\ell_1},
\end{align}
if $m_2\leq m_1$ (otherwise the values of of the upper and lower limits are swapped).

The angle between the normal to the zero set $\Sigma_P$ and the linear flow on the torus $\T^2$ is given by
\begin{equation}
\cos\theta=-\frac{\left(\frac{\omega_1\sqrt{m_2}}{\sin^2\varphi_1}+\frac{\omega_2\sqrt{m_1}}{\sin^2\varphi_2}\right)}{\| \omega \|\sqrt{\left(\frac{m_2}{\sin^4\varphi_1}+\frac{m_1}{\sin^4\varphi_2}\right)}}, \quad (\varphi_1,\varphi_2)\in\Sigma_P,
\end{equation}
and so $\cos\theta<0$ on $\Sigma_P$ (the orbit is nowhere tangent at $\Sigma_P$).
We can compute the Barra-Gaspard measure $\nu_P$ on $\Sigma_P$ of Theorem~\ref{thm:BG}. Using the same steps outlined in~\cite{Keating03} we get (with a slight abuse of notation)
\begin{equation}
\label{eq:BGmeas_KirDir}
    d\nu_P=\frac{1}{2 \pi  (\omega_1+\omega_2)}\left(\omega_1  \frac{\sqrt{m_1 m_2}}{m_1\sin ^2\varphi_1+m_2 \cos ^2\varphi_1}+\omega_2\right)d\varphi_1,\quad \varphi_1\in[0,2\pi),
\end{equation}
see Fig.~\ref{fig:subBG1}.
The Ces\`aro mean is computed from Theorem~\ref{thm:BG} as a space average on $\Sigma_P$ with respect to the measure $\nu_P$ in~\eqref{eq:BGmeas_KirDir}:
\begin{eqnarray}\label{cesmeanlim}
       \lim_{E\to +\infty}\frac{\sum_{E_{\kappa}\leq E} \mathcal{L}(\psi_{E_{\kappa}})}{\#\{E_{\kappa}\leq E\}}&=&\int_{\Sigma_P}\frac{\left(m_1\sin^2\varphi_1+m_2\cos^2\varphi_1\right)\ell_2-m_1\ell_1}{\left(m_1\sin^2\varphi_1+m_2\cos^2\varphi_1\right)\ell_2+m_1\ell_1}d\nu_P \nonumber \\
       &=&\frac{\sqrt{m_2}\ell_2-\sqrt{m_1}\ell_1}{\sqrt{m_2}\ell_2+\sqrt{m_1}\ell_1}.
\end{eqnarray}
See Fig.~\ref{fig:subLeaning0}.
\begin{rem}
   When the masses are equal, the spectral function drastically simplifies:
\begin{equation}
f_P(\varphi_1,\varphi_2)=\sqrt{m_1}\sin\left(\varphi_1+\varphi_2\right),\quad\text{ if $m_1=m_2$}.
\end{equation}
\end{rem} 
\subsection{A ring with Kirchhoff b.c.\ at two junctions}
\label{sec:KirRing}
\begin{equation}\label{eq:Ring}
    \begin{cases}
        \psi(0^-)=\psi(0^+),\\\\
           \psi(-\ell_1)=\psi(\ell_2),\\\\
           \dfrac{1}{m_1}\psi'(0^-)- \dfrac{1}{m_2}\psi'(0^+)=0, \\\\
        -\dfrac{1}{m_1}\psi'(-\ell_1)+ \dfrac{1}{m_2}\psi'(\ell_2)=0.
    \end{cases}
\end{equation}
These boundary conditions correspond to the topology of a ring (see Figure~\ref{fig:pictorialrepresentation}). 
In this case the projection $P=\ketbra{u}{u}+\ketbra{v}{v}$ (in Dirac notation) has rank $r=2$, with $u=\frac{1}{\sqrt{2}}(0,1,1,0)^T$ and $v=\frac{1}{\sqrt{2}}(1,0,0,1)^T$. The quasi-periodic component of the spectral function is
\begin{eqnarray}
f_P(\varphi_1,\varphi_2)&=&2\sqrt{m_1m_2}(1-  \cos\varphi_1 \cos\varphi_2) +(m_1+m_2) \sin \varphi_1 \sin \varphi_2 \label{eq:SF_KirRing}\\
                                       &=&4 \frac{(\sqrt{m_2}t_2+\sqrt{m_1}t_1)\sqrt{m_1}t_2+\sqrt{m_2}t_1)}{\left(1+t_1^2\right) \left(1+t_2^2\right)} \quad \textrm{{ (for $\varphi_1, \varphi_2 \neq \pi)$}}    \label{eq:SF_KirRing_halftan},
\end{eqnarray}
where $t_1=\tan(\varphi_1/2)$ and $t_2=\tan(\varphi_2/2)$.

Using Assumption~\ref{assump:1}, it is immediate to verify that $f_P(\varphi_1,\varphi_2)\neq0$ if $\varphi_1=\pi$ or $\varphi_2=\pi$. Therefore, the zero set $\Sigma_P$ of the spectral function is given by the zero set of the numerator of~\eqref{eq:SF_KirRing_halftan}. We see that, if $m_1\neq m_2$, then $ \Sigma_P=\Sigma_{P,1}\cup\Sigma_{P,2}$ where 
\begin{equation}
\Sigma_{P,1}=\left\{ (\varphi_1, \varphi_2) \in \T^2: t_2=-\sqrt{\frac{m_1}{m_2}}t_1\right\}, \quad \Sigma_{P,2}=\left\{(\varphi_1, \varphi_2) \in \T^2: t_2=-\sqrt{\frac{m_2}{m_1}}t_1\right\}.
\end{equation}
Since $0$ is not an eigenvalue of $H_{\mathbb{I}-2P}$, the set $\mathcal{K}_0= \emptyset$ and the spectrum of $H_{\mathbb{I}-2P}$ is 
\begin{eqnarray}
\sigma\left( H_{\mathbb{I}-2P} \right)&=&  \bigcup_{j=1}^2 \left \{E_{\kappa}=\frac{\hbar^2}{2}\kappa^2 : \kappa\in \mathcal{K}_j  \right\}.                                       \end{eqnarray}
For $\kappa \in \mathcal{K}_1$ the eigenfunctions are
\begin{equation}
\psi_{E_{\kappa}}(x)=
\frac{ \cos \left(\frac{\kappa \omega_1}{2}\right)}{ \sin \left(\kappa \omega_1\right)} \sin\left(\kappa\sqrt{m_1}\left(x+\frac{\ell_1}{2}\right)\right)\chi_{I_1}(x)-\frac{ \cos \left(\frac{\kappa \omega_2}{2}\right)}{ \sin \left(\kappa \omega_2\right)} \sin\left(\kappa\sqrt{m_2}\left(x-\frac{\ell_2}{2}\right)\right)\chi_{I_2}(x),
\end{equation}
while for $\kappa \in \mathcal{K}_2$ the eigenfunctions are
\begin{equation}
\psi_{E_{\kappa}}(x)=
\frac{ \sin \left(\frac{\kappa \omega_1}{2}\right)}{ \sin \left(\kappa \omega_1\right)} \cos\left(\kappa\sqrt{m_1}\left(x+\frac{\ell_1}{2}\right)\right)\chi_{I_1}(x)+\frac{ \sin \left(\frac{\kappa \omega_2}{2}\right)}{ \sin \left(\kappa \omega_2\right)} \cos\left(\kappa\sqrt{m_2}\left(x-\frac{\ell_2}{2}\right)\right)\chi_{I_2}(x).
\end{equation}
The corresponding leaning is, for $\kappa \to +\infty$,
\begin{equation}
    \label{eq:Leaning_Ring}
 \mathcal{L}(\psi_{E_{\kappa}})=
 \begin{cases}
     \frac{\left[m_1(1-\cos(\omega_1 \kappa))+m_2(1+\cos(\omega_2 \kappa))\right]\ell_2-m_1\ell_1}{\left[m_1(1-\cos(\omega_1 \kappa))+m_2(1+\cos(\omega_2 \kappa))\right]\ell_2-m_1\ell_1}+O\left(\frac{1}{\kappa}\right)&\text{if $\kappa \in \mathcal{K}_1$}\\\\
    \frac{\left[m_1(1+\cos(\omega_1 \kappa))+m_2(1-\cos(\omega_2 \kappa))\right]\ell_2-m_1\ell_1}{\left[m_1(1+\cos(\omega_1 \kappa))+m_2(1-\cos(\omega_2 \kappa))\right]\ell_2-m_1\ell_1}+O\left(\frac{1}{\kappa}\right)&\text{if $\kappa \in \mathcal{K}_2$}
 \end{cases}.
\end{equation}
Therefore, as in the previous case, the leaning densely fills the interval with endpoints
\begin{align}
\label{eq:suplimring}
\limsup_{\kappa \to + \infty} \mathcal{L}(\psi_{E_{\kappa}})&=\frac{\ell_2-\ell_1}{\ell_2+\ell_1},\\
\label{eq:inflimring}
\liminf_{\kappa \to + \infty} \mathcal{L}(\psi_{E_{\kappa}})&=\frac{m_2\ell_2-m_1\ell_1}{m_2\ell_2+m_1\ell_1},
\end{align}
if $m_2\leq m_1$ (otherwise the values of of the upper and lower limits are swapped). 

The angles between the orbit of the flow and the zero sets $\Sigma_{P,1}$ and $\Sigma_{P,2}$ are given by,
\begin{equation}
\cos\theta_1=\frac{\left(\frac{\omega_1\sqrt{m_1}}{1+\cos\varphi_1}+\frac{\omega_2\sqrt{m_2}}{1+\cos\varphi_2}\right)}{\| \omega\|\sqrt{\frac{m_1}{\left({1+\cos\varphi_1}\right)^2}+\frac{m_2}{\left({1+\cos\varphi_2}\right)^2}}}, \quad (\varphi_1,\varphi_2) \in \Sigma_{P,1},
\end{equation}
\begin{equation}
\cos\theta_2=
\frac{\left(\frac{\omega_1\sqrt{m_2}}{1+\cos\varphi_1}+\frac{\omega_2\sqrt{m_1}}{1+\cos\varphi_2}\right)}{\| \omega\|\sqrt{\frac{m_2}{\left({1+\cos\varphi_1}\right)^2}+\frac{m_1}{\left(1+\cos \varphi_2\right)^2}}}, \quad (\varphi_1,\varphi_2) \in \Sigma_{P,2}
\end{equation}
and so $\cos\theta_j>0$ on $\Sigma_{P,j}$, $j=1,2$, (the orbit is nowhere tangent at $\Sigma_P$). 
The Barra-Gaspard measure $\nu_P$ is supported on the two components with density
\begin{equation}
\label{eq:BGmeas_KirRing}
\frac{1}{4 \pi  (\omega_1+\omega_2)}\times
    \begin{cases}
     \left(\omega_1\frac{2  \sqrt{m_1m_2}}{(m_2-m_1) \cos \varphi_1+(m_1+m_2)}+\omega_2\right),\quad&\text{on $\Sigma_{P,1}$}\\\\
     \left(\omega_1\frac{2  \sqrt{m_1m_2}}{(m_1-m_2) \cos \varphi_1+(m_1+m_2)}+\omega_2\right),\quad&\text{on $\Sigma_{P,2}$}.
    \end{cases}
    \end{equation}
     With a slight abuse of notation we write
  \begin{equation}
    \label{eq:BGmeas_KirRing_2}
    d\nu_P=\frac{1}{2 \pi  (\omega_1+\omega_2)} \left(\omega_1\frac{2 \sqrt{m_1m_2} (m_1+m_2)}{(m_1+m_2)^2-(m_1-m_2)^2 \cos ^2\varphi_1}+\omega_2\right)d\varphi_1, \quad \varphi_1\in[0,2\pi),
\end{equation}
see Fig.~\ref{fig:subBG2}.
The Ces\`aro mean is computed from Theorem~\ref{thm:BG} as a space average on $\Sigma_P=\Sigma_{P,1} \cup \Sigma_{P,2}$ with respect to the measure $\nu_P$ in~\eqref{eq:BGmeas_KirRing}:
\begin{eqnarray}\label{cesmeanlim}
       \lim_{E\to +\infty}\frac{\sum_{E_{\kappa}\leq E} \mathcal{L}(\psi_{E_{\kappa}})}{\#\{E_{\kappa}\leq E\}}&=&\int_{\Sigma_{P,1}}\frac{\left[m_1(1-\cos\varphi_1)+m_2(1+\cos\varphi_2)\right]\ell_2-m_1\ell_1}{\left[m_1(1-\cos\varphi_1)+m_2(1+\cos\varphi_2)\right]\ell_2-m_1\ell_1} \, d\nu_{P,1}+ \nonumber \\
       && +\int_{\Sigma_{P,2}}\frac{\left[m_1(1+\cos\varphi_1)+m_2(1-\cos\varphi_2)\right]\ell_2-m_1\ell_1}{\left[m_1(1+\cos\varphi_1)+m_2(1-\cos\varphi_1)\right]\ell_2-m_1\ell_1} \, d\nu_{P,2} \nonumber \\
       &=&\frac{\sqrt{m_2}\ell_2-\sqrt{m_1}\ell_1}{\sqrt{m_2}\ell_2+\sqrt{m_1}\ell_1}.
\end{eqnarray}
See Fig.~\ref{fig:subLeaning1}.
\begin{rem}
   When the masses are equal, the spectral function  simplifies:
\begin{equation}
f_P(\varphi_1,\varphi_2)=2m(1- \cos(\varphi_1+\varphi_2)),\quad\text{ if $m_1=m_2$}.
\end{equation}
\end{rem} 
\subsection{A graph with a single pendant vertex}
\label{sec:pendant}
Imposing zero Kirchhoff b.c.\ at three vertices (here $0^-, 0^+$, and $\ell_2$), and zero Dirichlet b.c.\ at the fourth one (here $-\ell_1)$, we get: 
\begin{equation}\label{eq:Pendant}
    \begin{cases}
    \psi(0^-)=\psi(0^+)= \psi(\ell_2),\\\\
    \psi(-\ell_1)  =0,\\ \\
   \dfrac{1}{m_1}\psi'(0^-)- \dfrac{1}{m_2}\psi'(0^+)+ \dfrac{1}{m_2}\psi'(\ell_2)=0.
    \end{cases}
\end{equation}
These b.c.\ correspond to the topology of a graph with one pendant vertex (see Figure~\ref{fig:pictorialrepresentation}). We emphasize that, unlike all the other cases considered in above, this topology is not symmetric with respect to the exchange of the intervals $I_1$ and $I_2$.
In this case  $P=\ketbra{u}{u}$ (in Dirac notation) is the rank $r=1$ orthogonal projection onto the span of $u=\frac{1}{\sqrt{3}}(0,1,1,1)^T$. The quasi-periodic component of the spectral function is
\begin{eqnarray}
\label{eq:SF_KirPend}
\qquad f_P(\varphi_1,\varphi_2)&=& \sin \left(\frac{\varphi_2}{2}\right) \left[2 \sqrt{m_1} \sin \varphi_1 \sin \left(\frac{\varphi_2}{2}\right)-\sqrt{m_2} \cos \varphi_1 \cos \left(\frac{\varphi_2}{2}\right)\right] \\ 
&=& \frac{t_2}{(1+t_1^2)(1+t_2^2)}\left[4 \sqrt{m_1} t_1 t_2-\sqrt{m_2} \left(1-t_1^2\right)\right] \quad  \textrm{ (for $\varphi_1, \varphi_2\neq\pi)$}, \nonumber
 \end{eqnarray}
where $t_1=\tan(\varphi_1/2)$ and $t_2=\tan(\varphi_2/2)$. In this case $\Sigma_P=\Sigma_{P,1}\cup \Sigma_{P,2}$ where
\begin{equation}
\label{eq:zeroset_KirPend_halftan}
\Sigma_{P,1}=\{(\varphi_1,0): \varphi_1 \in [0,2\pi) \}, \quad  \Sigma_{P,2}= \left\{(\varphi_1, \varphi_2)\in \mathbb{T}^2: t_2=\sqrt{\frac{m_2}{m_1}}\frac{1-t_1^2}{4t_1}\right\}.
\end{equation}  
Since $0$ is not an eigenvalue of $H_{\mathbb{I}-2P}$, the set $\mathcal{K}_0= \emptyset$ and the spectrum of $H_{\mathbb{I}-2P}$ is 
\begin{eqnarray}
\sigma\left( H_{\mathbb{I}-2P} \right)&=&  \bigcup_{j=1}^2 \left \{E_{\kappa}=\frac{\hbar^2}{2}\kappa^2 : \kappa\in \mathcal{K}_j  \right\}.
                                        \end{eqnarray}
For $\kappa \in \mathcal{K}_1=\frac{2\pi}{\omega_2}\Z \setminus\{ 0\}$, the eigenfunctions are 
\begin{equation}
 \psi_{E_{\kappa}}(x)=  \sin\left(\kappa\sqrt{m_2} x\right)\chi_{I_2} (x),
 \end{equation}
while for $\kappa \in \mathcal{K}_2$,  the eigenfunctions are 
 \begin{equation}
\psi_{E_{\kappa}}(x)= \displaystyle\frac{\sin\left(\kappa\sqrt{m_1}(x+\ell_1)\right)}{\sin\left(\kappa\omega_1\right)}\chi_{I_1}(x)+\frac{\sin\left(\kappa\sqrt{m_2}x\right)-\sin\left(\kappa\sqrt{m_2}(x-\ell_2)\right)}{\sin\left(\kappa\omega_2\right)}\chi_{I_2}(x).
 \end{equation}
The corresponding leaning is, for $\kappa\to +\infty$,
\begin{equation}
    \label{eq:Leaning_Pendant}
 \mathcal{L}(\psi_{E_{\kappa}})=
 \begin{cases}
     1 &\text{if $\kappa \in \mathcal{K}_1$}\\\\
    \frac{(4m_1+m_2)-(4m_1-m_2)\cos(2\omega_1\kappa)-8m_1\ell_1}{(4m_1+m_2)-(4m_1-m_2)\cos(2\omega_1\kappa)+8m_1\ell_1} +O\left(\frac{1}{\kappa}\right) &\text{if $\kappa \in \mathcal{K}_2$}
 \end{cases}.
\end{equation}
In this case, there is a subsequence of eigenfunctions completely localised in $I_2$ (with leaning $1$), and the leaning of the other eigenfunctions densely fills the interval  with endpoints  
\begin{align}
\label{eq:upper_limit_pend}
\limsup_{\kappa \to + \infty} \mathcal{L}(\psi_{E_\kappa})&=\frac{\ell_2-\ell_1}{\ell_2+\ell_1},\\
    \label{eq:lower_limit_pend}
\liminf_{\kappa \to + \infty} \mathcal{L}(\psi_{E_\kappa})&=\frac{m_2\ell_2-4m_1\ell_1}{m_2\ell_2+4m_1\ell_1},
\end{align}
if $m_2\leq m_1$ (otherwise the values of the upper and lower limits are swapped).

The angles $\theta_1$ and $\theta_2$ between the orbit of the flow and the zero sets $\Sigma_{P,1}$ and $\Sigma_{P,2}$ respectively, are given by
\begin{equation}
\cos\theta_1=\frac{\omega_2}{\| \omega\|},
\end{equation}
for $(\varphi_1,\varphi_2) \in \Sigma_{P,1}$ and
\begin{equation}
\cos\theta_2=\frac{ 4 \omega_1 \sqrt{m_1 m_2}+4m_1 \omega_2 \sin ^2\varphi_1+m_2 \omega_2 \cos ^2\varphi_1}{\| \omega\|\sqrt{4 \left(m_2^2-16 m_1^2\right) \cos (2 \varphi_1)+48 m_1^2+(m_2-4m_1)^2 \cos (4 \varphi_1)+136 m_1 m_2+3 m_2^2}},
\end{equation}
for $(\varphi_1,\varphi_2) \in \Sigma_{P,2}$, so $\cos\theta_j>0$ on $\Sigma_{P,j}$, $j=1,2$, hence the orbit is nowhere tangent at $\Sigma_P$.
The Barra-Gaspard measure is supported on the two components with density
\begin{equation}
\label{eq:BGmeas_KirPend}
\frac{1}{4 \pi  (\omega_1+\omega_2)}\times
    \begin{cases}
         \omega_2\quad&\text{on $\Sigma_{P,1}$}\\\\
    \omega_1\frac{8\sqrt{m_1m_2}}{4m_1+m_2-(4m_1-m_2)\cos(2\varphi_1)}+\omega _2,\quad&\text{on $\Sigma_{P,2}$}.
    \end{cases}
    \end{equation}
    With a slight abuse of notation, we write
\begin{equation}
    \label{eq:BGmeas_Pend_2}
    d\nu_P=\frac{1}{2\pi\left(\omega_1+\omega_2\right)}\left(\omega_1\frac{4\sqrt{m_1m_2}}{4m_1+m_2-(4m_1-m_2)\cos(2\varphi_1)}+\omega_2\right)d\varphi_1,\quad\varphi_1\in[0,2\pi),
\end{equation}
see Fig.~\ref{fig:subBG3}. 
We can compute the Ces\`aro mean 
\begin{eqnarray}
       \lim_{E\to +\infty}\frac{\sum_{E_{\kappa}\leq E} \mathcal{L}(\psi_{E_{\kappa}})}{\#\{E_{\kappa}\leq E\}}&=&\int_{\Sigma_{P,1}}  d\nu_{P,1}+ \nonumber \\
       && +\int_{\Sigma_{P,2}}\frac{(4m_1+m_2)-(4m_1-m_2)\cos(2\varphi_1)-8m_1\ell_1}{(4m_1+m_2)-(4m_1-m_2)\cos(2\varphi_1)+8m_1\ell_1}  \, d\nu_{P,2} \nonumber \\
       &=&\frac{\omega_2/2}{\omega_1+\omega_2}+\frac{\omega_2/2-\omega_1}{\omega_1+\omega_2} \nonumber \\
       &=&\frac{\omega_2-\omega_1}{\omega_2+\omega_1} \nonumber \\
       &=&\frac{\sqrt{m_2}\ell_2-\sqrt{m_1}\ell_1}{\sqrt{m_2}\ell_2+\sqrt{m_1}\ell_1}. \nonumber
\end{eqnarray}
See Fig.~\ref{fig:subLeaning2}.
\begin{rem}
   When the masses are equal, the spectral function simplifies:
\begin{equation}
f_P(\varphi_1,\varphi_2)=\frac{\sqrt{m_1}}{2}\left( 2 \sin \varphi_1-2 \sin \varphi_1\cos \varphi_2- \cos \varphi_1\sin \varphi_2\right) ,\quad\text{ if $m_1=m_2$}.
\end{equation}
\end{rem} 
\subsection{A rose graph}
\label{sec:Rose}
Finally we consider the case of Kirchhoff b.c.\ at the four vertices:
\begin{equation}
\label{eq:bc_rose}
\begin{cases}
\psi(-\ell_1)=\psi(0^-)=\psi(0^+)=\psi(\ell_2),\\\\
    -\dfrac{1}{m_1}\psi'(-\ell_1)  +\dfrac{1}{m_1}\psi'(0^-)- \dfrac{1}{m_2}\psi'(0^+)+ \dfrac{1}{m_2}\psi'(\ell_2)=0.
    \end{cases}
\end{equation}
This corresponds to a rose graph (a graph with two loops on one single vertex, see  Figure~\ref{fig:pictorialrepresentation}).
The b.c.~\eqref{eq:bc_rose} are scale-free with $P=\ketbra{u}{u}$ (in Dirac notation) of rank $r=1$ with $u=\frac{1}{2}(1,1,1,1)^T$. The quasi-periodic component of the spectral function is
\begin{eqnarray}
\label{eq:SF_KirRose}
f_P(\varphi_1,\varphi_2)&=& \sqrt{m_1}\sin\varphi_1 (1-\cos \varphi_2)+\sqrt{m_2}\sin\varphi_2  (1-\cos \varphi_1) \\
                                      &=& 4\frac{t_1t_2}{(1+t_1^2)(1+t_2^2)}(\sqrt{m_2}t_1+\sqrt{m_1}t_2) \quad \textrm{ (for $\varphi_1, \varphi_2\neq\pi)$},
\end{eqnarray}
where $t_1=\tan(\varphi_1/2)$ and $t_2=\tan(\varphi_2/2)$. In this case $\Sigma_P=\Sigma_{P,1}\cup \Sigma_{P,2} \cup \Sigma_{P,3}$ where
\begin{equation}
\Sigma_{P,1}=\{(0,\varphi_2): \varphi_2\in [0,2\pi) \}, \quad \Sigma_{P,2}=\{(\varphi_1,0): \varphi_1\in [0,2\pi) \},
\end{equation}    
and
\begin{equation}
\Sigma_{P,3}= \left\{(\varphi_1, \varphi_2)\in \mathbb{T}^2:t_2=-\sqrt{\frac{m_2}{m_1}}t_1\right\}.
\end{equation}
Since $0$ is an eigenvalue of $H_{\mathbb{I}-2P}$ the set $\mathcal{K}_0=\{0\}$ and the spectrum of $H_{\mathbb{I}-2P}$ is 
\begin{eqnarray}
\sigma\left( H_{\mathbb{I}-2P} \right)&=&  \bigcup_{j=0}^3 \left \{E_{\kappa}=\frac{\hbar^2}{2}\kappa^2 : \kappa\in \mathcal{K}_j  \right\}.
\end{eqnarray}
For $\kappa=0$ the eigenfunctions are same as \eqref{eq:zeromod}, for $\kappa \in \mathcal{K}_1=\frac{2\pi}{\omega_1}\Z  \setminus\{ 0\}$, the eigenfunctions are \begin{equation}
 \psi_{E_{\kappa}}(x)=  \sin\left(\kappa\sqrt{m_1} x\right)\chi_{I_1} (x),
 \end{equation}
with leaning $\mathcal{L}(\psi_{E_{\kappa}})=-1$, while for $\kappa \in \mathcal{K}_2=\frac{2\pi}{\omega_2}\Z \setminus\{ 0\}$, the eigenfunctions are 
\begin{equation}
 \psi_{E_{\kappa}}(x)=  \sin\left(\kappa\sqrt{m_2} x\right)\chi_{I_2} (x),
 \end{equation}
 with leaning $\mathcal{L}(\psi_{E_{\kappa}})=1$.
For $\kappa \in \mathcal{K}_3$, the eigenfunctions are 
\begin{equation}
 \psi_{E_{\kappa}}(x)= \displaystyle\frac{ \cos \left(\kappa\sqrt{m_1}\left(x+\frac{\ell_1}{2}  \right)\right)}{\cos\left(\frac{\kappa\omega _1}{2}\right)}\chi_{I_1}(x)+ \frac{ \cos \left(\kappa\sqrt{m_2}\left(x-\frac{\ell_2}{2}  \right)\right)}{\cos\left(\frac{\kappa\omega _2}{2}\right)}\chi_{I_2}(x),
\end{equation}
 and the corresponding leaning is
  \begin{equation}
  \mathcal{L}(\psi_{E_{\kappa}})=\frac{ m_1 \ell_2(1+\cos (\omega_1\kappa))+m_2\ell_2 m_2 (1-\cos (\omega_1\kappa))-2 m_1\ell_1}{m_1 \ell_2(1+\cos (\omega_1\kappa))+m_2\ell_2 m_2 (1-\cos (\omega_1\kappa))+2 m_1\ell_1 } +O\left(\frac{1}{\kappa}\right), \quad \textrm{as $\kappa \to + \infty$}.
 \end{equation}
In this case, there are two subsequences of eigenfunctions completely localised in $I_1$ (with leaning $-1$) and $I_2$ (with leaning $1$); all other eigenfunctions have leaning that densely fills the interval with endpoints 
\begin{align}
\label{eq:upper_limit_rose}
\limsup_{\kappa \to + \infty} \mathcal{L}(\psi_{E_\kappa})&=\frac{\ell_2-\ell_1}{\ell_2+\ell_1},\\
    \label{eq:lower_limit_rose}
\liminf_{\kappa \to + \infty} \mathcal{L}(\psi_{E_\kappa})&=\frac{m_2\ell_2-4m_1\ell_1}{m_2\ell_2+4m_1\ell_1},
\end{align}
if $m_2\leq m_1$ (otherwise the values of of the upper and lower limits are swapped). 
The Barra-Gaspard measure $\nu_P$ is supported on the three components with density
\begin{equation}
\label{eq:BGmeas_KirRose}
\frac{1}{4 \pi  (\omega_1+\omega_2)}\times
    \begin{cases}
         \omega_1\quad&\text{on $\Sigma_{P,1}$}\\\\
     \omega_2\quad&\text{on $\Sigma_{P,2}$}\\\\
     \omega_1\frac{2  \sqrt{m_1m_2}}{(m_1-m_2) \cos \varphi_1+(m_1+m_2)}+\omega_2 \quad&\text{on $\Sigma_{P,3}$}.
    \end{cases}
    \end{equation}
The Ces\`aro mean, computed using~\eqref{eq:BGmeas_KirRose}, is again given by~\eqref{cesmeanlim}:
\begin{eqnarray}
       \lim_{E\to +\infty}\frac{\sum_{E_{\kappa}\leq E} \mathcal{L}(\psi_{E_{\kappa}})}{\#\{E_{\kappa}\leq E\}}&=&-\int_{\Sigma_{P,1}}  d\nu_{P,1} +\int_{\Sigma_{P,2}} d\nu_{P,2} \nonumber \\
       &&+\int_{\Sigma_{P,3}}  \frac{ m_1 \ell_2(1+\cos \varphi_1)+m_2\ell_2 m_2 (1-\cos \varphi_1)-2 m_1\ell_1}{m_1 \ell_2(1+\cos \varphi_1)+m_2\ell_2 m_2 (1-\cos \varphi_1)+2 m_1\ell_1 } d\nu_{P,3} \nonumber \\
       &=&\frac{\sqrt{m_2}\ell_2-\sqrt{m_1}\ell_1}{\sqrt{m_2}\ell_2+\sqrt{m_1}\ell_1},
\end{eqnarray}
see Fig.~\ref{fig:subLeaning3}.
\begin{rem}
   When the masses are equal, the spectral function becomes:
\begin{equation}
f_P(\varphi_1,\varphi_2)=\sqrt{m_1}\left(\sin\varphi_1+\sin\varphi_2-\sin(\varphi_1+\varphi_2)\right),\quad\text{ if $m_1=m_2$}.
\end{equation}
\end{rem}

 \begin{figure}[th]
  \centering
  \begin{subfigure}[b]{0.49\textwidth}
    \includegraphics[width=\textwidth]{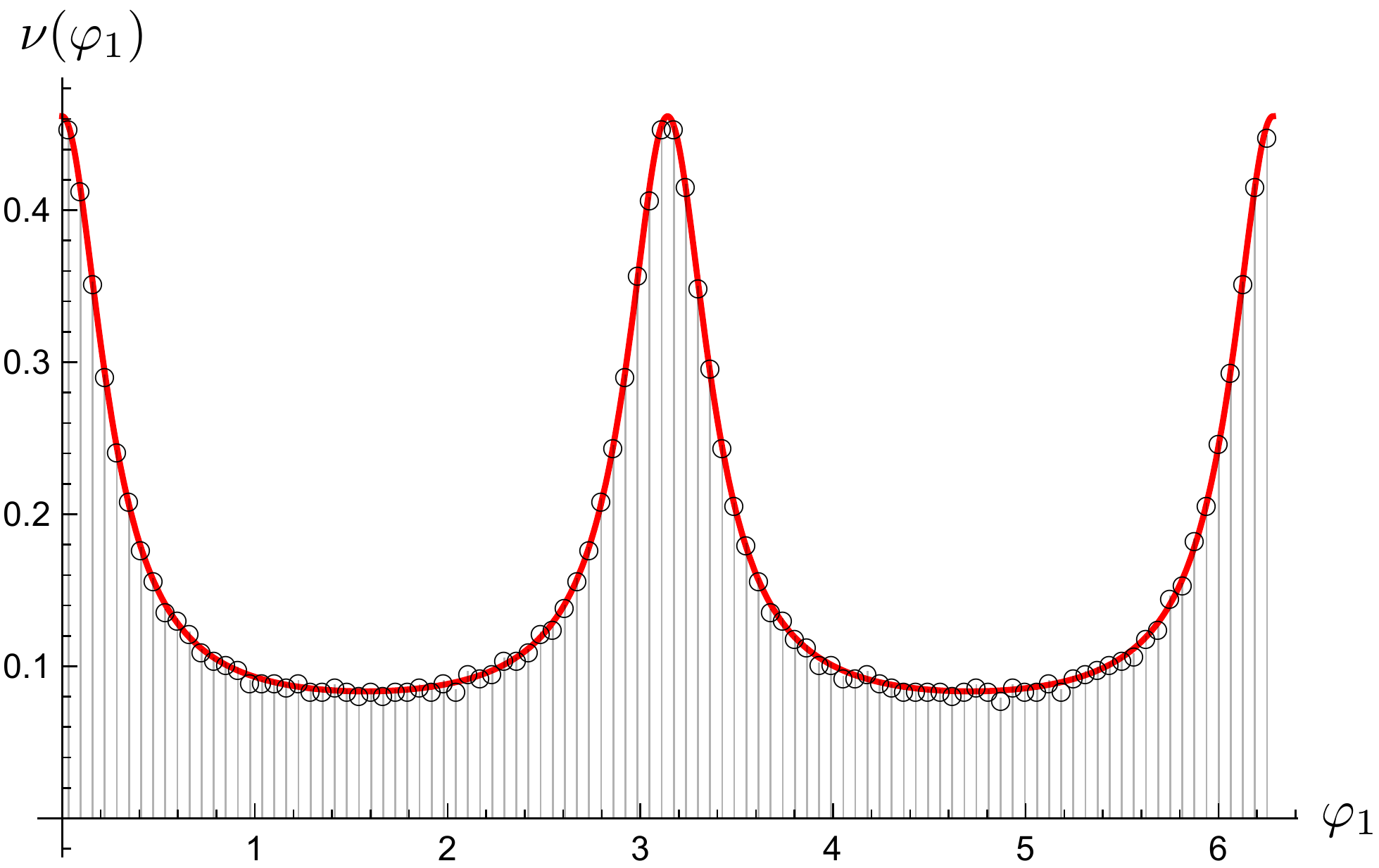}  
    \caption{The segment.}
    \label{fig:subBG1}
  \end{subfigure}
  \hfill
  \begin{subfigure}[b]{0.49\textwidth}
    \includegraphics[width=\textwidth]{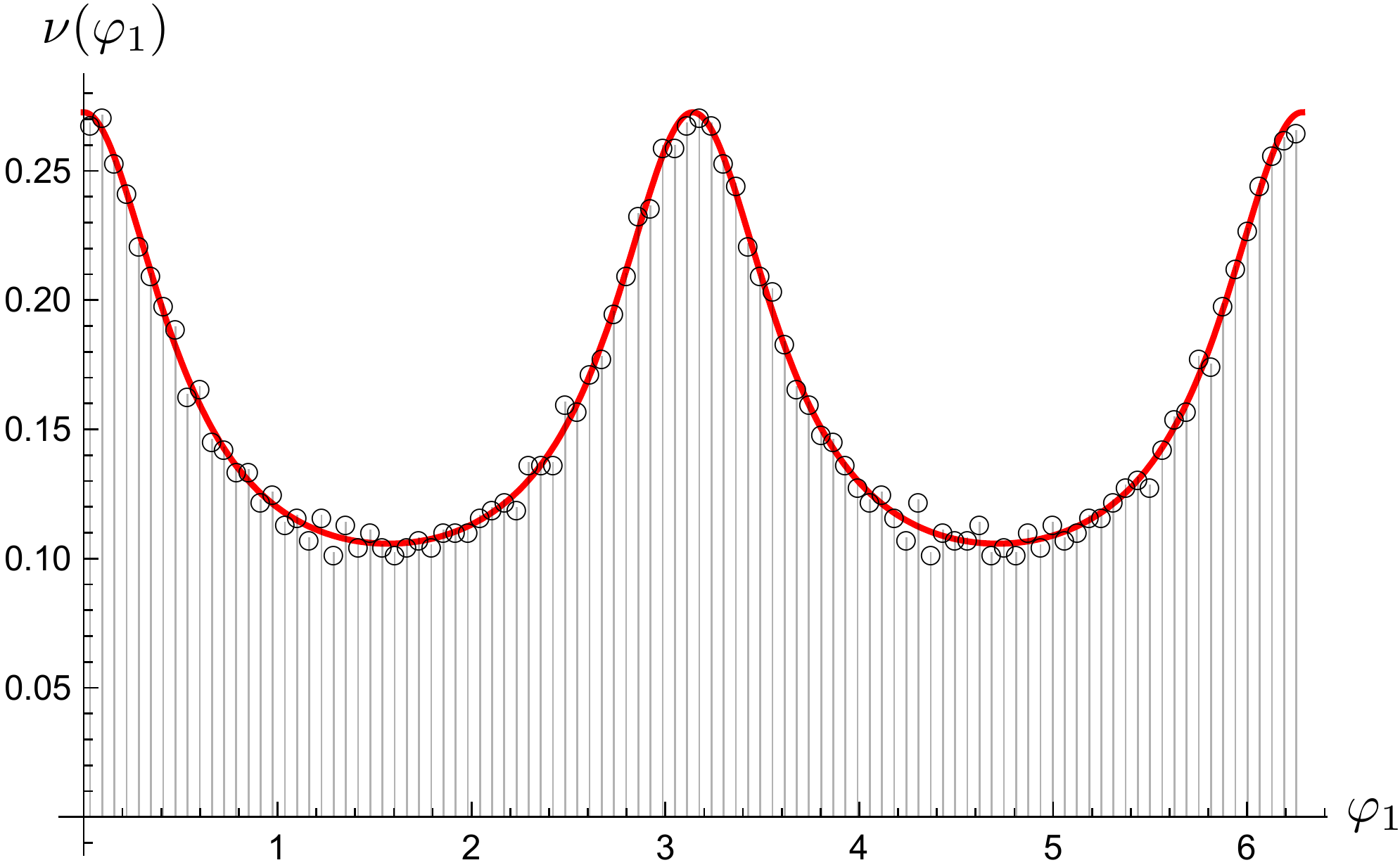}
    \caption{The ring.}
    \label{fig:subBG2}
  \end{subfigure}
    \hfill
  \begin{subfigure}[b]{0.49\textwidth}
    \includegraphics[width=\textwidth]{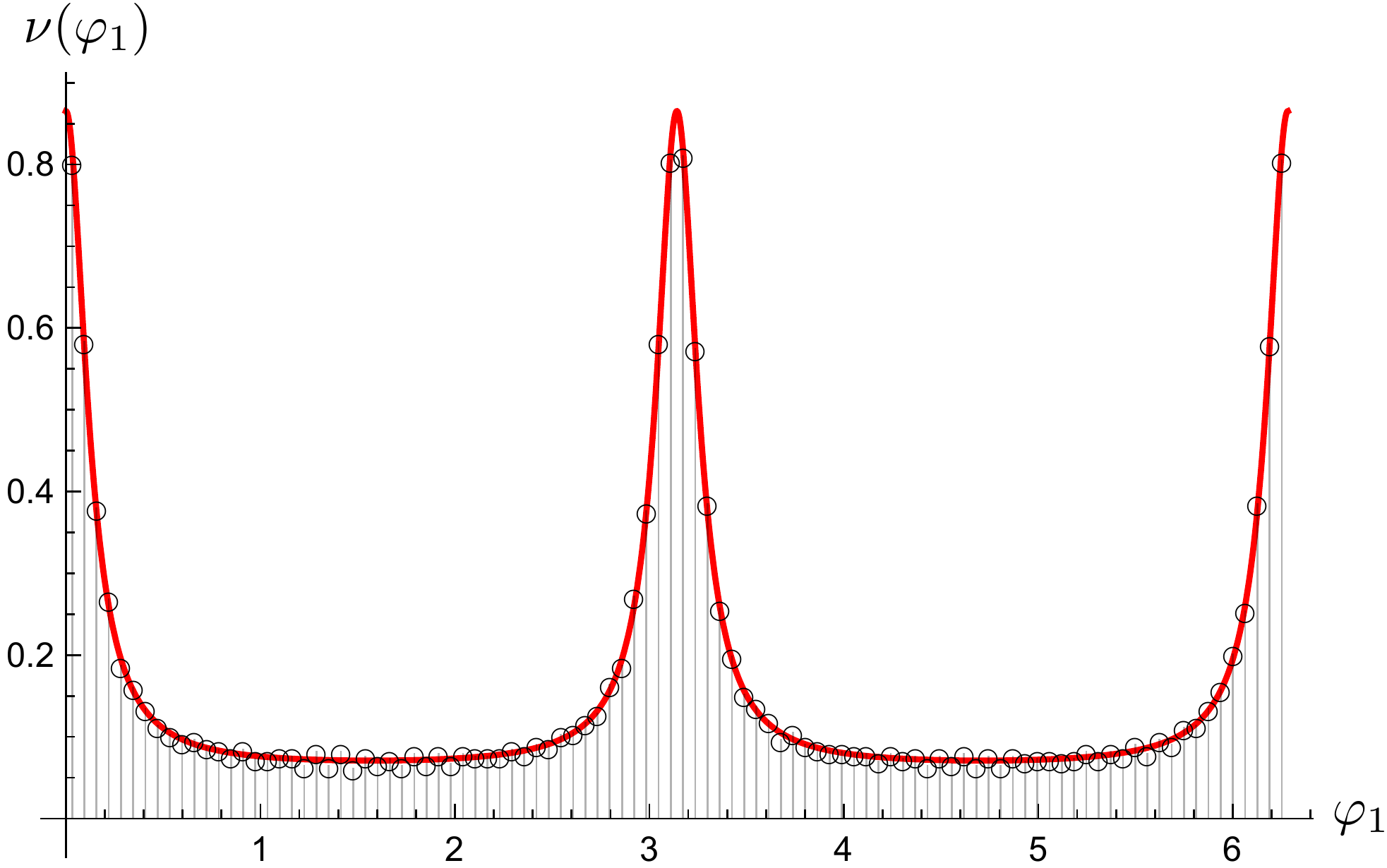}
    \caption{The single pendant.}
    \label{fig:subBG3}
  \end{subfigure}
  \hfill
  \begin{subfigure}[b]{0.49\textwidth}
    \includegraphics[width=\textwidth]{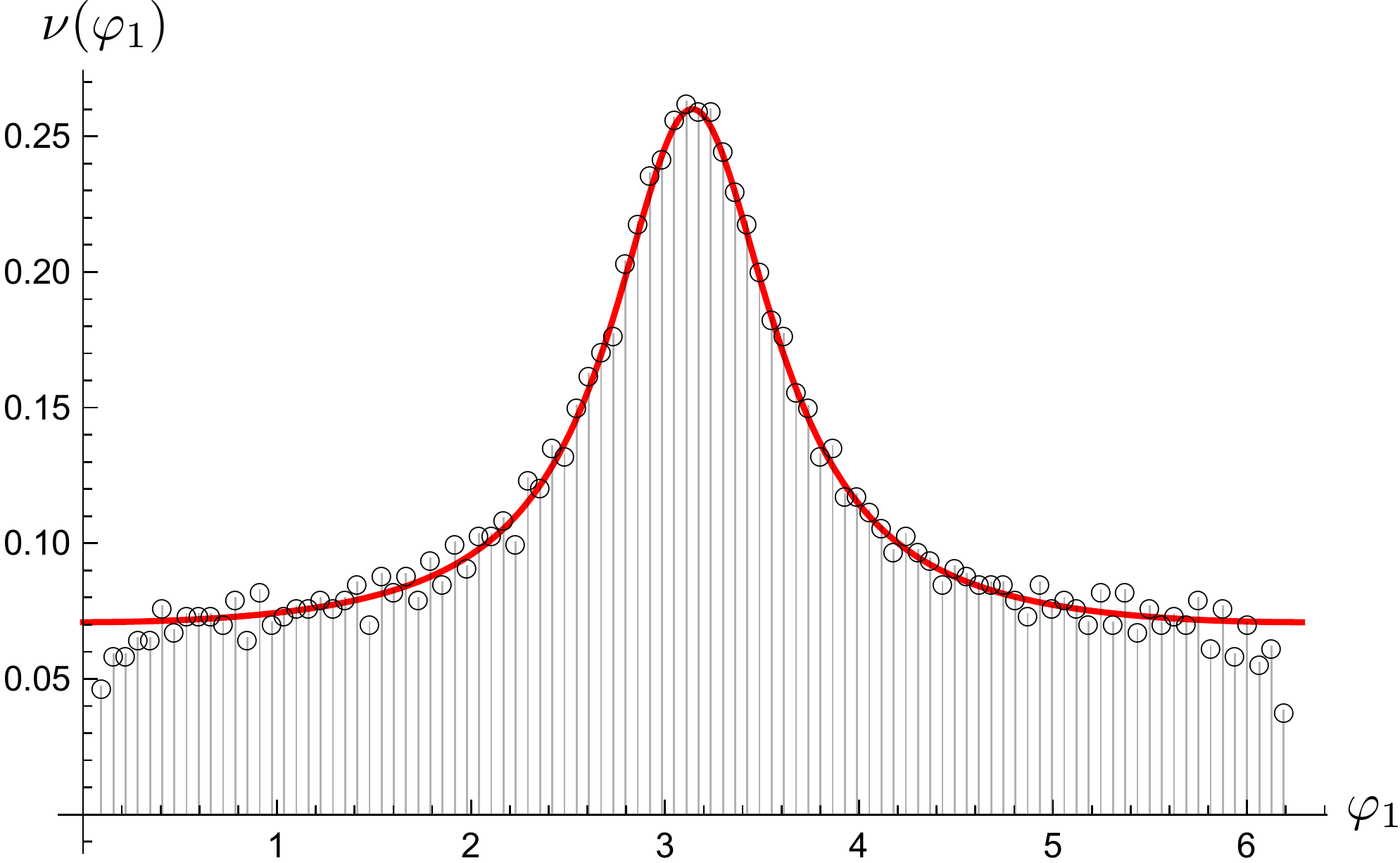}
    \caption{The rose.}
    \label{fig:subBG4}
  \end{subfigure}
  \caption{Barra-Gaspard measures $d\nu_P(\varphi_1)$ from~\eqref{eq:BGmeas_KirDir}-\eqref{eq:BGmeas_KirRing_2}-\eqref{eq:BGmeas_Pend_2}-\eqref{eq:BGmeas_KirRose} (see Figure~\ref{fig:pictorialrepresentation} for the corresponding pictorial representations). The dots are the numerical points $\omega_1\kappa\mod 2\pi$, where $\kappa \neq 0$ are such that  $f_P(\omega_1 \kappa, \omega_2 \kappa)=0$, with $f_P$ in~\eqref{eq:SF_KirDir}-\eqref{eq:SF_KirRing}-\eqref{eq:SF_KirPend}-\eqref{eq:SF_KirRose}. For the rose we omit in the plot the delta at $\varphi_1=0$. The parameters are the same of Fig.~\ref{fig:SF}.}
  \label{fig:BG}
\end{figure}

 \begin{figure}[h]
  \centering
    \begin{subfigure}{0.45\textwidth}
    \includegraphics[width=\textwidth]{Leaning_KirkGG.pdf}
    \caption{The segment.}
           \label{fig:subLeaning0}
      \end{subfigure}
          \hfill
  \begin{subfigure}{0.45\textwidth}
    \includegraphics[width=\textwidth]{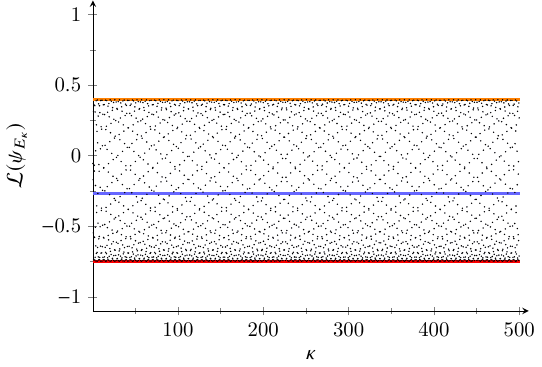}
    \caption{The ring.}
    \label{fig:subLeaning1}
  \end{subfigure}
    \hfill
  \begin{subfigure}{0.45\textwidth}
    \includegraphics[width=\textwidth]{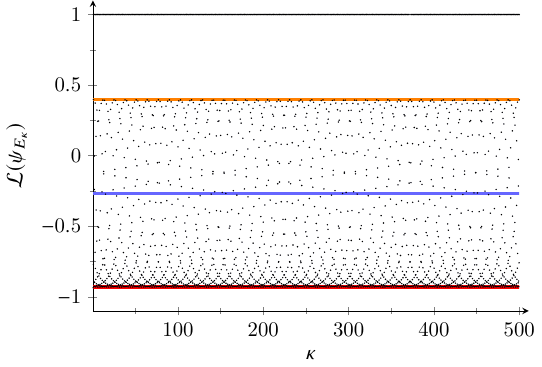}
    \caption{The single pendant.}
    \label{fig:subLeaning2}
  \end{subfigure}
  \hfill
  \begin{subfigure}{0.45\textwidth}
    \includegraphics[width=\textwidth]{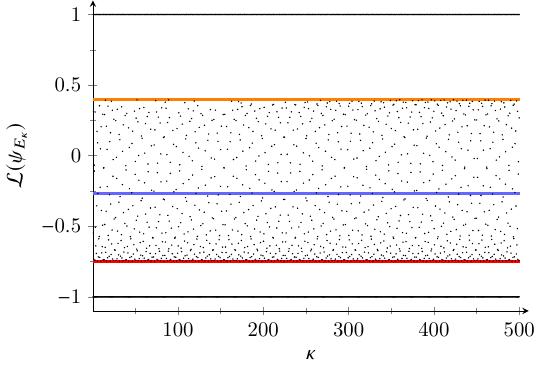}
    \caption{The rose.}
    \label{fig:subLeaning3}
  \end{subfigure}
  \caption{Numerical values (black dots) of the leaning  $\mathcal{L}(\psi_{E_{\kappa}})$ versus $\kappa$. The upper (in orange) and the lower (in red) 
 limits of the band are the one computed in Sections~\ref{sec:KirDir}-~\ref{sec:KirRing}-~\ref{sec:pendant}-~\ref{sec:Rose}. The Ces\`aro limit in blue is also displayed. The parameters are the same of Fig.~\ref{fig:SF}.}
  \label{fig:LeaningAll}
\end{figure}

\section{Semiclassical structure of the eigenfunctions}
\label{sec:semiclassical}

   \subsection{Wigner functions}
   \label{sec:Wigner}
   To understand the semiclassical behaviour of the eigenfunctions, it is natural to consider their associated Wigner functions. The idea is that, in the semiclassical limit, the Wigner function usually offers a phase space portrait of the quantum-to-classical correspondence.
  \subsection{General structure}  
   The Wigner function $W\psi$ of $\psi\in L^2(\R)$ is the function  
   \begin{equation}
       W\psi(x,p)=\frac{1}{2\pi\hbar}\int_{\R}\overline{\psi\left(x-\frac{y}{2}\right)}\psi\left(x+\frac{y}{2}\right)\rme^{-\ii py/\hbar}dy,
   \end{equation}
   on the phase space $\mathbb{R}_x\times\mathbb{R}_p$. This function is real-valued and captures, in a distributional sense, the quantum probability density in phase space.  It satisfies the normalization identity 
     \begin{equation}
     \int_{\R_x\times\R_p}   W\psi(x,p)dxdp=\|\psi\|_{L^2(\R)}.
   \end{equation}
   Moreover $W\psi\in L^2(\R_x\times\R_p)\cap C_0(\R_x\times\R_p)$, where $C_0$ is the space of continuous functions vanishing at infinity,  and $ \| W\psi\|_{L^2(\R_x\times\R_p)}=\frac{1}{\sqrt{2\pi\hbar}}\|\psi\|^2_{L^2(\R)}$. The reader can consult~\cite{Lions93,Folland} for further details.
 
  In our setting, let $\psi=\psi_1\chi_{I_1}+\psi_2\chi_{I_2}$ be a wavefunction on $\Omega=I_1\cup I_2$.  The Wigner function naturally decomposes as
   \begin{equation}
       W\psi=W_{11}+W_{22}+W_{12}+W_{21},
   \end{equation}
   where each term corresponds to contributions from auto- and cross-correlations between the components:
   \begin{equation}
       W_{mn}(x,p)=\frac{1}{2\pi\hbar}\int_{R_{mn}(x)}\overline{\psi_m\left(x-\frac{y}{2}\right)}\psi_n\left(x+\frac{y}{2}\right)\rme^{-\ii py/\hbar}dy,\quad m,n=1,2.
   \end{equation}

\par

    Here, the  integration regions in the variable $y$ (see Fig.~\ref{fig:int_regions}) reflect the finite support of the wave components and depend explicitly on $x$:
   \begin{align}
   \label{eq:int_regions}
       R_{11}(x)&=\left\{y\in\R\colon 2\left|x+\frac{\ell_1}{2}\right|-\ell_1\leq y\leq \ell_1-2\left|x+\frac{\ell_1}{2}\right|\right\} \quad &\text{for $x\in[-\ell_1,0]$}\\ \notag
        R_{22}(x)&=\left\{y\in\R\colon 2\left|x-\frac{\ell_2}{2}\right|-\ell_2\leq y\leq \ell_2-2\left|x-\frac{\ell_2}{2}\right|\right\} \quad &\text{for $x\in[0,\ell_2]$}\\ \notag
        R_{12}(x)&=\left\{y\in\R\colon 2\left|x\right|\leq y\leq (\ell_1+\ell_2)-2\left|x-\frac{\ell_2-\ell_1}{2}\right|\right\} \quad &\text{for $x\in\left[-\frac{\ell_1}{2},\frac{\ell_2}{2}\right]$}\\ \notag
        R_{21}(x)&=\left\{y\in\R\colon -(\ell_1+\ell_2)+2\left|x-\frac{\ell_2-\ell_1}{2}\right| \leq y\leq -2\left|x\right| \right\} \quad &\text{for $x\in\left[-\frac{\ell_1}{2},\frac{\ell_2}{2}\right]$}.
   \end{align}
   Note that $R_{12}(x)=-R_{21}(x)$ for all $x\in\left[-\frac{\ell_1}{2},\frac{\ell_2}{2}\right]$, and hence 
     \begin{equation}
       W\psi=W_{11}+W_{22}+2\operatorname{Re} W_{12}.
   \end{equation}

   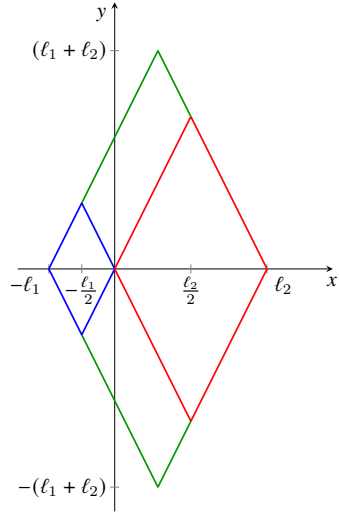
\begin{figure}
   \centering
\begin{tikzpicture}[scale=0.8]
  \begin{axis}[
    axis lines=middle,
    width=16cm,
    height=10cm,
    xmin=-4, xmax=9,
    ymin=-10, ymax=11,
    samples=400,
    domain=-3:7,
    xtick=\empty,
    extra x ticks={-\Lone, -\Lone/2,\Ltwo/2,\Ltwo},
    extra x tick labels={\empty},
    ytick=\empty,
    extra y ticks={-\Lone-\Ltwo,\Lone+\Ltwo},
    extra y tick labels={\empty},
    axis equal image,
    clip=false,
    enlargelimits=false,
  ]

  \pgfmathsetmacro{\Lone}{exp(1)}
  \pgfmathsetmacro{\Ltwo}{2*pi}
  \pgfmathsetmacro{\xshift}{(\Ltwo - \Lone)/2}
  \pgfmathsetmacro{\sumL}{\Lone + \Ltwo}

  \addplot[thick, blue, domain=-\Lone:0] 
    { \Lone - 2*abs(x + \Lone/2) };

  \addplot[thick, red, domain=0:\Ltwo] 
    { \Ltwo - 2*abs(x - \Ltwo/2) };

  \addplot[thick, green!60!black, domain=-\Lone/2:\Ltwo/2] 
    { \sumL - 2*abs(x - \xshift)  };

  \addplot[ thick, blue, domain=-\Lone:0] 
    { -(\Lone - 2*abs(x + \Lone/2)) };

  \addplot[ thick, red, domain=0:\Ltwo] 
    { -(\Ltwo - 2*abs(x - \Ltwo/2)) };

  \addplot[thick, green!60!black, domain=-\Lone/2:\Ltwo/2] 
    { -(\sumL - 2*abs(x - \xshift)) };

  \node at (axis cs:-\Lone,0) [below left] {$-\ell_1$};
  \node at (axis cs:-\Lone/2,0) [below] {$-\frac{\ell_1}{2}$};
  \node at (axis cs:\Ltwo/2,0) [below] {$\frac{\ell_2}{2}$};
  \node at (axis cs:\Ltwo,0) [below right] {$\ell_2$};

  \node at (axis cs:0,{-\sumL}) [ left] {$-(\ell_1 +\ell_2)$};
  \node at (axis cs:0,{\sumL}) [ left] {$(\ell_1 + \ell_2)$};
    \node at (axis cs:0,10.5) [left] {$y$};
      \node at (axis cs:9,0) [below] {$x$};
  \end{axis}
\end{tikzpicture}
  \caption{Integration regions~\eqref{eq:int_regions} for the computation of $W\psi$.}
  \label{fig:int_regions}
\end{figure}

The Wigner function on quantum graphs was considered in~\cite{Barra2001} without the cross terms $W_{12}$ and $W_{21}$.
\subsection{Wigner function of $\psi_{E_{\kappa}}$ and its semiclassical limits}
Our main interest lies in the Wigner functions of the normalised eigenfunctions $\psi_{E_{\kappa}}$ of $H_U$ for all possible $U \in \mathcal{U}(4)$. 
After computing the resulting integrals explicitly, we get the following.
\begin{prop}
\label{prop:Wigner} Let $\psi_{E_{\kappa}}$ be an eigenfunction of $H_U$ as in~\eqref{eq:eigenfunction}. Denote $k_j=\kappa\sqrt{m_j}$, for $j=1,2$. Then,

\begin{align}
\label{eq:Wigner}
    \notag W\psi_{E_{\kappa}}(x,p)&=  \\
     \frac{1}{\pi}\chi_{I_1}(x) a_{11}(x) 
    &\left[ |c_1|^2\operatorname{S}_{11}\left(x,p-\hbar k_1\right) 
    + |d_1|^2\operatorname{S}_{11}\left(x,p+\hbar k_1\right) 
    + 2\Re c_1\overline{d_1}\rme^{2\ii k_1x}\operatorname{S}_{22}\left(x,p\right) \right] \notag \\
    + \frac{1}{\pi} \chi_{I_2}(x)a_{22}(x) 
    &\left[ |c_2|^2\operatorname{S}_{22}\left(x,p-\hbar k_2\right) 
    + |d_2|^2\operatorname{S}_{22}\left(x,p+\hbar k_2\right) 
    + 2\Re c_2\overline{d_2}\rme^{2\ii k_2x}\operatorname{S}_{11}\left(x,p\right) \right] \notag \\
    + \frac{2}{\pi}\chi_{I_{12}(x)} a_{12}(x)&\Re  
    \biggl[ 
        {c_1}\overline{c_2} \rme^{\ii(k_1-k_2)x} \rme^{\ii \widetilde{a}_{12}(x)\left(p-\hbar \frac{k_1+k_2}{2}\right)} 
\operatorname{S}_{12}\left(x,p-\hbar \frac{k_1+k_2}{2}\right) \notag \\
    &\quad\!+ {c_1}\overline{d_2} \rme^{\ii(k_1+k_2)x} \rme^{\ii\widetilde{a}_{12}(x)\left(p-\hbar \frac{k_1-k_2}{2}\right)} \operatorname{S}_{12}\left(x,p-\hbar \frac{k_1-k_2}{2}\right) \notag \\
    &\quad\!+ {d_1}\overline{c_2} \rme^{-\ii(k_1+k_2)x} \rme^{\ii\widetilde{a}_{12}(x)\left(p+\hbar \frac{k_1-k_2}{2}\right)} \operatorname{S}_{12}\left(x,p+\hbar \frac{k_1-k_2}{2}\right) \notag \\
    &\quad\! + {d_1}\overline{d_2} \rme^{-\ii(k_1-k_2)x} e^{i\widetilde{a}_{12}(x)\left(p+\hbar \frac{k_1+k_2}{2}\right)} \operatorname{S}_{12}\left(x,p+\hbar \frac{k_1+k_2}{2}\right)
    \biggr],
\end{align}
where $I_1=(-\ell_1,0)$, $I_2=(0,\ell_2)$, $I_{12}=\left(-\frac{\ell_1}{2},\frac{\ell_2}{2}\right)$, 
\begin{align}
   a_{11}(x)&=\frac{1}{ \hbar} \left(\ell_1-2\left|x+\frac{\ell_1}{2}\right| \right),&
a_{22}(x)&= \frac{1}{ \hbar} \left(\ell_2-2\left|x-\frac{\ell_2}{2}\right| \right),\\
   a_{12}(x)&=\frac{1}{ \hbar} \left(\frac{\ell_1+\ell_2}{2}-\left|x-\frac{\ell_2-\ell_1}{2}\right|-|x|\right),&
   \widetilde{a}_{12}(x)&=\frac{1}{ \hbar} \left(\frac{\ell_1+\ell_2}{2}-\left|x-\frac{\ell_2-\ell_1}{2}\right|+|x| \right),
\end{align}
and 
\begin{equation}
    \operatorname{S}_{ij}(x,p):=\frac{\sin\left(a_{ij}(x)p\right)}{a_{ij}(x)p}=
    \sinc(a_{ij}(x)p),\quad i,j=1,2.
\end{equation}
\end{prop}
\begin{figure}
    \centering
    \includegraphics[width=.75\linewidth]{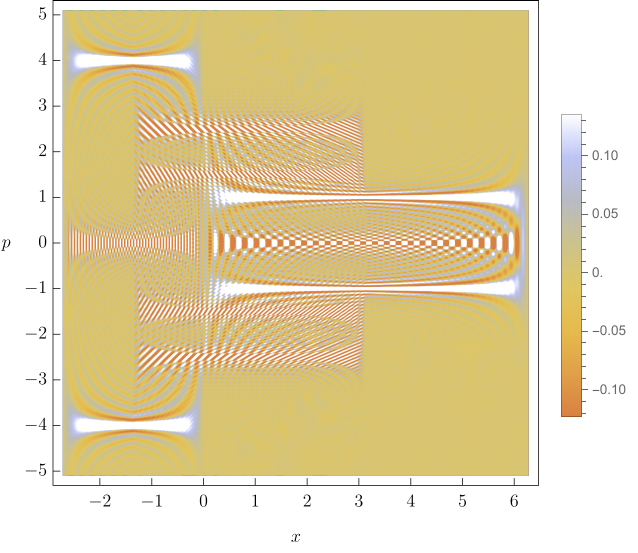}
    \caption{Wigner function $W\psi_{E_{\kappa}}$ for an eigenfunction of $H_U$, see~\eqref{eq:Wigner}. Here $U$ corresponds to Kirchhoff b.c.\ at the junction $x=0$ and Dirichlet b.c.\ at the pendant vertices, and  $\kappa$ corresponds to the $74$th energy level of $H_{\mathbb{I}-2P}$. The parameters are the same of Fig.~\ref{fig:SF}, $E=1/2$, and $\hbar$ is fixed by the condition $\hbar \kappa=\sqrt{2E}.$}
    \label{fig:Wigner}
\end{figure}
The Wigner function~\eqref{eq:Wigner} contains  oscillatory contributions modulated by $\sinc$ functions. The auto-correlation terms contain $\sinc$ functions in momentum, centered at $\pm \hbar k_1$ for $x\in I_1$, and $\pm \hbar k_2$, for $x\in I_2$. The $\sinc$ centered at $p=0$ is oscillatory in position $x$.  The cross term, $2\operatorname{Re}W_{12}$,  supported in the interval $x\in[-\ell_1/2,\ell_2/2]$, exhibits interference between the two components, built from shifted $\sinc$ functions and phase oscillations in $x$. See Fig.~\ref{fig:Wigner}.
\par

The study of the limiting behaviour of the Wigner function~\eqref{eq:Wigner} for $\hbar \to 0$ is rather complex and 
beyond the purposes of this article. However, we want to point out some features that can be observed even pointwise around the mass-discontinuity point $x=0$: we consider the semiclassical regime of~\eqref{eq:Wigner}, i.e.  we fix a macroscopic energy $E > 0$ and take the limit $\hbar \to 0$ while keeping $\hbar \kappa = \sqrt{2E}$ fixed (large wavenumber) of $\hbar W\psi_{E_{\kappa}}(\hbar u,p)$. This corresponds to zooming in near the mass-discontinuity point $x = 0$ on a spatial scale of order $\hbar$, while simultaneously increasing the energy to explore the high-frequency behaviour of the system. (For similar calculations see~\cite{Cunden25}.) 
\par
We remark that the coefficients $c_1,d_1,c_2,d_2 \in \C$ in Proposition \ref{prop:Wigner} depend on $\kappa$, hence we can write $c_1(\kappa),d_1(\kappa),c_2(\kappa),d_2(\kappa)$. Suppose that, upon extraction of a subsequence $\{\kappa_n\}_{n\geq1}$, the vector of coefficients converges,
\begin{equation}
\left(c_1(\kappa_n), d_1(\kappa_n),c_2(\kappa_n),d_2(\kappa_n)\right)\to \left(\mathfrak{c}_1, \mathfrak{d}_1,\mathfrak{c}_2,\mathfrak{d}_2\right),\quad \text{as $n\to+\infty$}.
\end{equation}
This limit can (and in fact does) depend on the chosen subsequence. Then, along that subsequence, the scaled Wigner function converges pointwise to a limiting function
\begin{equation}
    \lim_{\substack{\hbar\to0,\kappa\to+\infty\\ \hbar\kappa=\sqrt{2E}}}\hbar W\psi_{E_{\kappa}}(\hbar u,p)=W_E(u,p),
\end{equation}
with
\begin{align}
\label{eq:limit_Wigner}
    &W_E(u,p)=  \\
    + \frac{2u}{\pi} \chi_{\mathbb{R}_-}(u)
    &\left[ |\mathfrak{c}_1|^2\sinc\left(2u(p-p_1)\right) 
    + |\mathfrak{d}_1|^2\sinc\left(2u(p+p_1)\right) 
    + 2\Re \left(\mathfrak{c}_1\overline{\mathfrak{d}_1}e^{2ip_1u}\right)\sinc\left(2up\right) \right] \notag \\
     -\frac{2u}{\pi} \chi_{\mathbb{R}_+}(u)
    &\left[ |\mathfrak{c}_2|^2\sinc\left(2u(p-p_2)\right) 
    + |\mathfrak{d}_2|^2\sinc\left(2u(p+p_2)\right) 
    + 2\Re \left(\mathfrak{c}_2\overline{\mathfrak{d}_2}e^{2ip_2u}\right)\sinc\left(2up\right) \right] \notag \\
    \qquad\qquad - \frac{2|u|}{\pi}&\operatorname{Re}  
    \biggl[ 
        {\mathfrak{c}_1}\overline{\mathfrak{c}_2} e^{+i(p_1-p_2)u} 
 \sinc\left(2|u|\left(p-\frac{p_1+p_2}{2}\right)\right) \notag \\
    &\quad+ {\mathfrak{c}_1}\overline{\mathfrak{d}_2} e^{+i(p_1+p_2)u}  \sinc\left(2|u|\left(p- \frac{p_1-p_2}{2}\right)\right) \notag \\
    & \quad+ {\mathfrak{d}_1}\overline{\mathfrak{c}_2} e^{-i(p_1+p_2)u} \sinc\left(2|u|\left(p+ \frac{p_1-p_2}{2}\right)\right) \notag \\
    &\quad+ {\mathfrak{d}_1}\overline{\mathfrak{d}_2} e^{-i(p_1-p_2)u} \sinc\left(2|u|\left(p+ \frac{p_1+p_2}{2}\right)\right) 
    \biggr],\notag
\end{align}
where we denoted the classical momenta $p_i=\sqrt{2m_i E}$, for $i=1,2$.
\par

The structure of $W_E(u,p)$ captures the classical portrait in phase space, particularly the orbits traversing the junction point at $u=0$, with interference fringes appearing in the cross terms. The auto-correlation term contains $\sinc$ functions in momentum, centered at the classical momenta:   $p=0$ and $p=\pm p_1=\pm\sqrt{2m_1 E}$ on the left $u<0$, and $\pm p_2\pm\sqrt{2m_1 E}$ on the right $u>0$. The cross-correlation part includes $\sinc$ functions centered at momentum values $p = \pm(p_1 + p_2)/2$ and $p = \pm(p_1 - p_2)/2$, accompanied by oscillations in position. The amplitudes of all these components are determined by the limiting coefficients  $\mathfrak{c}_1, \mathfrak{d}_1,\mathfrak{c}_2,\mathfrak{d}_2$, which themselves depend on the specific  subsequence $\{ E_{\kappa_n}\}_{n \geq 1}$ of energy levels.

\subsection{Non-Unique Semiclassical Measures and Toral Parametrization}
Suppose that $H_U$ is a scale-free kinetic energy operator. To fix  ideas, consider the case of zero Kirchhoff b.c.\ at the junction $x=0$ (the matrix $U$ in~\eqref{eq:U_KirDir}). See the initial discussion in Sec~\ref{sec:glimpse},  and Sec.~\ref{sec:KirDir} for relevant formulae. 
\par

When the frequencies $\omega_1,\omega_2$ are nonresonant (Assumption~\ref{assump:1}), by Theorem~\ref{prop:mean}, every point $(\varphi_1,\varphi_2)\in\Sigma_P\subset\mathbb{T}^2$ is an accumulation point for the set $\{(\omega_1\kappa,\omega_2\kappa) \mod 2\pi\colon f_P(\omega_1\kappa,\omega_2\kappa)=0\}$.  Therefore, for all $(\varphi_1,\varphi_2)\in\Sigma_P$ there exists a subsequence $(\kappa_n)_{n\geq1}$ of zeros of the spectral function such that $(\omega_1\kappa_n,\omega_2\kappa_n)  \to (\varphi_1,\varphi_2)$, as $n\to+\infty$, on $\mathbb{T}^2$. By continuity, the vector of coefficients (see~\eqref{eq:coeff_KirDir}) converges, as $n\to\infty$, to
\begin{equation}
 \left(\mathfrak{c}_1, \mathfrak{d}_1,\mathfrak{c}_2,\mathfrak{d}_2\right)=   \frac{1}{\sqrt{\frac{2\ell_1}{\sin^2\varphi_1}+\frac{2\ell_2}{\sin^2\varphi_2}}}\left(\frac{\rme^{\ii\varphi_1}}{\sin\varphi_1},-\frac{\rme^{-\ii \varphi_1}}{\sin\varphi_1},-\frac{\rme^{-\ii \varphi_2}}{\sin\varphi_2},\frac{\rme^{\ii \varphi_2}}{\sin\varphi_2}\right).
\end{equation}
Therefore, each point of the spectral curve on the torus  corresponds to a possible limit $W_E$ of the Wigner distribution~\eqref{eq:limit_Wigner}.
\par

We conclude that the sequence of Wigner functions $W\psi_{E_{\kappa}}$ admits \emph{infinitely many distinct semiclassical limits}, each corresponding to a different subsequence $\{E_{\kappa_n}\}_{n \geq 1}$ of energy levels. These limits are parametrized by points $(\varphi_1,\varphi_2)$ on a spectral curve $\Sigma_P$ embedded in the two-dimensional torus $\mathbb{T}^2$.  The emergence of multiple microlocal limits along distinct subsequences of eigenstates is a semiclassical echo of the discontinuous character of the quantum systems.

\section{Conclusions and outlook}
\label{sec:conclusions}
We introduced and studied the problem of a quantum particle in  dimension one with a jump-discontinuous mass. The model is naturally cast as a two-edges quantum graph with different masses on the edges and self-adjoint boundary conditions. We showed how the problem can be analysed in the case of scale-free b.c.\ and solved in detail the associated spectral problem for certain natural boundary conditions. The eigenfunctions spatial properties display erratic behaviour (measured for instance by the `leaning') that has no counterpart in the integrable case of constant mass. The discontinuity of the mass  makes the problem almost-integrable: too intricate to be pinned down, but not so complex to be chaotic or ergodic. 

Several intriguing questions remain open. Chief among them is whether the observed irregular dependence of the eigenfunctions on energy persists across all self-adjoint extensions — that is, for arbitrary boundary conditions at the discontinuities. A related issue concerns the universality of spectral statistics: are quantities such as the Ces\`aro mean, or other spectral averages, invariant under changes in boundary data (all our examples share the same value of the Ces\`aro mean of the leaning)? Concerning the quantum dynamics, the unitary evolution of a quantum particle with a discontinuous mass profile remains largely unexplored and poses significant challenges for future investigation.
 
 All this poses the question of studying  general quantum graphs (i.e. general metric graphs) with \emph{discontinuous mass}. 
   The next step is to solve the spectral problem on 
 two-edge quantum graphs with different mass on the edges  and {general} boundary conditions. This will be the subject of a subsequent paper~\cite{CGL}. Our findings indicate that understanding the semiclassical limit(s) of quantum graphs (one-dimensional quantum billiards) with discontinuous mass is still a challenging task. It is somewhat unsettling that, although we can define and study one-dimensional quantum billiards with jump-discontinuous mass, we find ourselves struggling to define -- let alone comprehend -- the classical motion of billiards with similar discontinuities in mass. This is yet another intriguing challenge in our ongoing pursuit to understand the classical-to-quantum boundary.

\section*{Acknowledgements} 
This research is partially supported by Gruppo Nazionale di Fisica Matematica GNFM-INdAM and by Istituto Nazionale di Fisica Nucleare INFN through the project QUANTUM. FDC and ML acknowledge the support from  PRIN 2022 project 2022TEB52W-PE1-
`The charm of integrability: from nonlinear waves to random matrices'.  ML acknowledge the support by PNRR MUR project CN00000013 `Italian National Centre on HPC, Big Data and Quantum Computing'. GG acknowledges financial support from PNRR MUR project PE0000023-NQSTI and from the University of Bari through the 2023-UNBACLE-0245516 grant. The authors thank Davide Lonigro for discussions at the initial stage of the project, and Paolo Facchi, Daniel Burgarth and Matteo Gallone for  valuable comments.

\appendix


\begin{thebibliography}{10}


        \bibitem{Anantharaman22} N. Anantharaman,
        \emph{Quantum Ergodicity and Delocalization of Schr\"odinger Eigenfunctions},
        Zurich Lectures in Advanced Mathematics, EMS Press, 2022.
	
	\bibitem{Angelone2024}
	G. Angelone, P. Facchi, M. Ligab\`o,
	\emph{Classical echoes of quantum boundary condition}
	Journal of Physics A: Mathematical and Theoretical {\bf 57}, 425304 (2024).
	
	\bibitem{Ares2020} 
	F. Ares, J. G. Esteve, F. Falceto,  A. Us\'on,
	\emph{Complex behavior of the density in composite quantum systems},
	Phys. Rev. B {\bf 102}, 165121 (2020).

    \bibitem{Arnold89} V. I. Arnold,
    \emph{Mathematical Methods of Classical Mechanics},
    Graduate Texts in Mathematics, Springer New York, NY, 1989.
	
	\bibitem{Asorey05} 
	M. Asorey, A. Ibort, G. Marmo, \emph{Global theory of quantum boundary conditions and topology change}, International Journal of Modern Physics A \textbf{20.05}, 1001-1025 (2005).
	
	\bibitem{Asorey15} 
	M. Asorey, A. Ibort, G. Marmo,
	\emph{The topology and geometry of self-adjoint and elliptic boundary conditions for Dirac and Laplace operators},
	International Journal of Geometric Methods in Modern Physics \textbf{12.06}, 1561007 (2015).

    \bibitem{Bagchi04} B. Bagchi, P. Gorain, C. Quesne, R. Roychoudhury,
    \emph{A general scheme for the effective-mass Schrodinger equation and the generation of the associated potentials},
    	Mod. Phys. Lett. A {\bf 19},  2765-2775 (2004).

        \bibitem{Barra2000} F. Barra, P. Gaspard,
        \emph{On the Level Spacing Distribution in Quantum Graphs},
        J. Stat. Phys. {\bf 101},  283-319 (2000).

        \bibitem{Barra2001} F. Barra, P. Gaspard,
        \emph{Transport and dynamics on open quantum graphs},
        Phys. Rev. E {\bf 65}, 016205 (2001).

        \bibitem{BenDaniel66} D. J. BenDaniel and C. B. Duke,
        \emph{Space-Charge Effects on Electron Tunneling},
        Phys. Rev. {\bf 152}, 683-692 (1966). 

        \bibitem{Berkolaiko2003}
G. Berkolaiko, J. P. Keating, B. Winn,
    \emph{Intermediate Wave Function Statistics},
    Phys. Rev. Lett. {\bf 91}, 134103 (2003).
    
	\bibitem{Berkolaiko2004}
	G. Berkolaiko, J. P. Keating, B. Winn,
	\emph{No Quantum Ergodicity for Star Graphs},
	Commun. Math. Phys. {\bf 250}, 259-285 (2004).

	\bibitem{Berkolaiko2013}
    G. Berkolaiko and P. Kuchment,
\emph{Introduction to Quantum Graphs},
American Mathematical Society, 2013.

        \bibitem{Blumel2001}
        R. Bl\"umel, P. M. Koch, L. Sirko,
        \emph{Ray-Splitting Billiards},
        Foundations of Physics {\bf 31(2)}, 269-281 (2001).

        \bibitem{Blumel2002}
         R. Bl\"umel, Y. Dabaghian,  R. V. Jensen
        \emph{Exact, convergent periodic-orbit expansions of individual energy eigenvalues
of regular quantum graphs},
        Phys. Rev. E {\bf 65}, 046222 (2002).
    
	\bibitem{Bolte09}
	J. Bolte, S. Endres,
	\emph{The Trace Formula for Quantum Graphs with General Self Adjoint Boundary Conditions},
	Ann. Henri Poincar\'e \textbf{10}, 189-223 (2009).
	
	\bibitem{Bruening08}
	J. Br\"{u}ning, V. Geyler, K. Pankrashkin,
	\emph{Spectra of self-adjoint extensions and applications to solvable Schr{\"o}dinger operators},
	Reviews in Mathematical Physics \textbf{20.01}, 1-70 (2008).

    \bibitem{Cunden18}
    F. D. Cunden, F. Mezzadri, N. O'Connell,
    \emph{Free fermions and the classical compact groups},
    J. Stat. Phys. {\bf 171}, 768-801 (2018).
    
    \bibitem{Cunden25}
    F. D. Cunden,  M. Ligab\`o, M. C. Susca,
    \emph{Truncated quantum observables and their semiclassical limit},
    J. Four. Anal. Appl. {\bf 31}, 46 (2025).
    
    \bibitem{CGL}
    F. D. Cunden, G. Gramegna,  M. Ligab\`o,
    \emph{Quantum Systems with jump-discontinuous mass. II}, in preparation.

    \bibitem{Dekar99} L. Dekar, L. Chetouani, T. F. Hammann,
    \emph{Wave function for smooth potential and mass step},
    Phys. Rev. A {\bf 59}, 107-112 (1999).

    \bibitem{Facchi18}
    P. Facchi, G. Garnero, M. Ligab\`o,
    \emph{Self-adjoint extensions and unitary operators on the boundary},
    Letters in Mathematical Physics {\bf 108}, 195 (2018).

    \bibitem{Folland}
G. B. Folland, 
\emph{Harmonic analysis in phase space.}
Princeton university press, 2016.


	\bibitem{Jakobson2015}
	D. Jakobson, Y. Safarov, A. Strohmaier, Y. Colin de Verdi\`ere,
	\emph{The semiclassical theory of discontinuous systems and ray-splitting billiards},
	American Journal of Mathematics {\bf 137 (4)}, 859-906 (2015).

    \bibitem{Kaplan01} L. Kaplan,
    \emph{Eigenstate structure in graphs and disordered lattices},
    Phys. Rev. E {\bf 64}, 036225 (2001).

\bibitem{Keating03} J. P. Keating, J. Marklof, B. Winn,
\emph{Value Distribution of the Eigenfunctions and Spectral Determinants of Quantum Star Graphs},
Commun. Math. Phys. {\bf 241}, 421-452 (2003).


    \bibitem{Koc05} R. Ko{\c c}, M. Koca, G. {\c S}ahino\v{g}lu,
    \emph{Scattering in abrupt heterostructures using a position dependent mass Hamiltonian},
    Eur. Phys. J. B {\bf 48}, 583-586 (2005).
    
	\bibitem{Kostrykin99}
	V. Kostrykin, R. Schrader,
	\emph{Kirchhoff's rule for quantum wires},
	J. Phys. A: Math. Gen. \textbf{32}, 595-630 (1999).
	
	\bibitem{Kostrykin00}
	V. Kostrykin, R. Schrader,
	\emph{Kirchhoff's rule for quantum wires. II: The inverse problem with possible applications to quantum computers},
	Fortschritte der Physik: Progress of Physics \textbf{48.8}, 703-716 (2000).

    \bibitem{Lions93} 
    P. L. Lions and T. Paul, 
    \emph{Sur les mesures de Wigner}, Revista Matematica Iberoamericana {\bf 9(3)}, 553-618 (1993).



    \bibitem{Pena17} J. J. Pe\~na, J. Morales, J. Garc\'ia-Ravelo, L. Arcos-D\'iaz,
    \emph{Schro\"odinger Equation with Position-Dependent Mass:
Staggered Mass Distributions},
International Journal of Physical and Mathematical Sciences {\bf 8}, 11 (2017).

\bibitem{ReedSimon}
M. Reed, B. Simon,
\emph{Methods of modern mathematical physics. II. Fourier analysis, self-adjointness.}
Academic Press [Harcourt Brace Jovanovich, Publishers], New York-London, 1975.


    \bibitem{Ross83}
    O. von Ross,
    \emph{Position-dependent effective mass in semiconductor theory},
    Phys. Rev. B {\bf 27}, 7547-7552 (1983).
    
\bibitem{Verdiere15}
Y. Colin de Verdi\`ere.
    \emph{Semi-Classical Measures on Quantum Graphs and the Gauss Map of the Determinant Manifold}, 
    Ann. Henri Poincar\'e {\bf16},  347-364 (2015).

    
\end{thebibliography}
\end{document}